\documentclass{article}

\usepackage{cmbright}
\usepackage{amssymb,amsmath,amsthm}
\usepackage{mathrsfs}
\usepackage[colorlinks,allcolors=blue]{hyperref}
\usepackage{graphicx}

\def\B{\mathscr B}
\def\C{\mathbb C}
\def\d{\mathrm d}
\def\D{\mathbb D}
\def\Drond{\mathscr D}
\def\dom{\mathcal D}
\def\Dt{\mathfrak D}

\def\f{\mathfrak f}
\def\F{\mathscr F}
\def\G{\mathcal G}
\def\H{\mathcal H}
\def\Hrond{\mathscr H}
\def\HH{\mathfrak H}
\def\h{\mathfrak h}
\def\K{\mathscr K}
\def\Kt{\mathfrak K}
\def\L{\mathcal L}
\def\M{\mathrm M}
\def\N{\mathbb N}
\def\NN{\mathrm N}
\def\O{\mathcal O}
\def\P{\mathcal P}
\def\R{\mathbb R}
\def\S{\mathbb S}
\def\T{\mathcal T}
\def\u{\mathfrak u}
\def\U{\mathcal U}
\def\Ut{\mathfrak U}
\def\v{\mathfrak v}
\def\V{\mathcal V}
\def\Xt{\mathfrak X}
\def\Z{\mathbb Z}

\def\diag{\mathop{\mathrm{diag}}\nolimits}
\def\e{\mathop{\mathrm{e}}\nolimits}
\def\im{\mathop{\mathrm{Im}}\nolimits}
\def\re{\mathop{\mathrm{Re}}\nolimits}
\def\Ran{\mathop{\mathrm{Ran}}\nolimits}
\def\sgn{\mathop{\mathrm{sgn}}\nolimits}
\DeclareMathOperator*{\slim}{s\hspace{0.1pt}-\hspace{0.1pt}lim}
\def\supp{\mathop{\mathrm{supp}}\nolimits}
\def\lone{\mathop{\mathrm{L}^1}\nolimits}
\def\ltwo{\mathop{\mathrm{L}^2}\nolimits}
\def\Oas{\mathop{{\mathcal O}_{\rm as}}\nolimits}
\def\pv{\mathop{\mathrm{p.v.}}\nolimits}
\def\arctanh{\mathop{\mathrm{arctanh}}\nolimits}
\def\csch{\mathop{\mathrm{csch}}\nolimits}
\def\sech{\mathop{\mathrm{sech}}\nolimits}
\def\sp{\mathop{\mathrm{span}}\nolimits}

\newtheorem{Theorem}{Theorem}[section]
\newtheorem{Remark}[Theorem]{Remark}
\newtheorem{Lemma}[Theorem]{Lemma}

\newtheorem{Proposition}[Theorem]{Proposition}

\newtheorem{Example}[Theorem]{Example}

\renewcommand{\theequation}{\arabic{section}.\arabic{equation}}

\begin{document}


\title{Discrete Laplacian in a half-space with a periodic surface potential I\,\!:
Resolvent expansions, scattering matrix, and wave operators}

\author{H.~S. Nguyen$^1$,
~S. Richard$^1$\footnote{Supported by the grant\emph{Topological invariants
through scattering theory and noncommutative geometry} from Nagoya University,
and by JSPS Grant-in-Aid for scientific research C no 18K03328, and on
leave of absence from Univ.~Lyon, Universit\'e Claude Bernard Lyon 1, CNRS UMR 5208,
Institut Camille Jordan, 43 blvd.~du 11 novembre 1918, F-69622 Villeurbanne cedex,
France.},
~R. Tiedra de Aldecoa$^2$\footnote{Partially supported by the Chilean Fondecyt Grant 1170008.}}

\date{\small}
\maketitle
\vspace{-1cm}

\begin{quote}
\emph{
\begin{enumerate}
\item[$^1$] Graduate school of mathematics, Nagoya University,
Chikusa-ku,\\Nagoya 464-8602, Japan
\item[$^2$] Facultad de Matem\'aticas, Pontificia Universidad Cat\'olica de Chile,\\
Av. Vicu\~na Mackenna 4860, Santiago, Chile
\item[]\emph{E-mails:} nguyen.ha.song@b.mbox.nagoya-u.ac.jp,
richard@math.nagoya-u.ac.jp,\\\phantom{\emph{E-mails:}~}rtiedra@mat.uc.cl
\end{enumerate}
}
\end{quote}


\begin{abstract}
We present a detailed study of the scattering system given by the Neumann Laplacian on
the discrete half-space perturbed by a periodic potential at the boundary. We derive
asymptotic resolvent expansions at thresholds and eigenvalues, we prove the continuity
of the scattering matrix, and we establish new formulas for the wave operators. Along
the way, our analysis puts into evidence a surprising relation between some properties
of the potential, like the parity of its period, and the behaviour of the integral
kernel of the wave operators.
\end{abstract}

\textbf{2010 Mathematics Subject Classification:} 81Q10, 47A40

\smallskip

\textbf{Keywords:} resolvent expansions, scattering matrix, wave operators, discrete
Laplacian, thresholds.

\tableofcontents

\section{Introduction and main results}\label{sec_intro}
\setcounter{equation}{0}

For the last 20 years, Schr\"odinger operators with potentials supported on lower
dimensional subspaces have been the subject of an intensive study motivated by both
physical applications and mathematical interest, see for example
\cite{BBP03,Cha00,CS00,FK04,Fra03,Fra04,JL00} and references therein. These systems
exhibit properties that are intermediate between the ones of standard scattering
systems (with potentials decaying in all space directions) and the ones of bulk
systems (with potentials having no specific space decay). A fundamental example of
such property, appearing in discrete and in continuous settings, is the presence of
surface states propagating along the lower dimensional subspace. Our goal is to
present a detailed study of these surface states from a $C^*$-algebraic point of view
for a two-dimensional system on the discrete lattice. In particular, we plan to
establish an index-type theorem relating the surface states to the scattering part of
the system, as it was done in various other contexts \cite{BSB,GP,Ric16,SB16}.
However, before any $C^*$-algebraic construction and prior to any index theorem, a lot
of analysis is needed. This is the subject of this first part of a series of two
papers.

The model that we consider is a simple and natural quantum system exhibiting surface
states. It is given by a Laplace operator on a discrete half-space, subject to a
periodic potential at the boundary. See Figure \ref{fig:model}. Despite its
simplicity, this model requires a non-trivial analysis, and exhibits some unexpected
properties. The model has already been studied, for instance in \cite{BBP03,Cha00},
but our paper contains more extensive results on scattering theory, presented within
an up-to-date framework.

Let us now give a description of our principal results. In the Hilbert space
$\H:=\ell^2(\Z\times\N)\cong\ell^2(\Z)\otimes\ell^2(\N)$, we consider the free
Hamiltonian
$$
H_0:=\Delta_\Z\otimes1+1\otimes\Delta_\NN,
$$
where $\Delta_\Z$ is the adjacency operator in $\ell^2(\Z)$ given by
$$
\big(\Delta_\Z\;\!\varphi\big)(x)
:=\varphi(x+1)+\varphi(x-1),\quad\varphi\in\ell^2(\Z),~x\in\Z,
$$
and where $\Delta_\NN$ is the discrete Neumann adjacency operator in $\ell^2(\N)$
given by
$$
\big(\Delta_\NN\;\!\phi\big)(n)=
\begin{cases}
2^{1/2}\;\!\phi(1) &\hbox{if $n=0$}\\
2^{1/2}\;\!\phi(0)+\phi(2) &\hbox{if $n=1$}\\
\phi(n+1)+\phi(n-1) &\hbox{if $n\ge2$,}
\end{cases}
\quad\phi\in\ell^2(\N),~n\in\N.
$$
As a full Hamiltonian, we consider the operator
$$
H:=H_0+V,
$$
where $V$ is the multiplication operator by a nonzero, periodic, real-valued function
with support on $\Z\times\{0\}$. In other words, we assume that there exists a nonzero
periodic function $v:\Z\to\R$ of period $N\in\N$, $N\ge2$, (the potential) such that
$$
(H\psi)(x,n)=(H_0\psi)(x,n)+\delta_{0,n}\;\!v(x)\;\!\psi(x,0),
\quad\psi\in\H,~x\in\Z,~n\in\N,
$$
with $\delta_{0,n}$ the Kronecker delta function. Note that the multiplication
operator $V$ associated to the potential $v$ is not a compact perturbation of $H_0$.

\begin{figure}\label{fig:model}
\centering
\includegraphics[width=300pt]{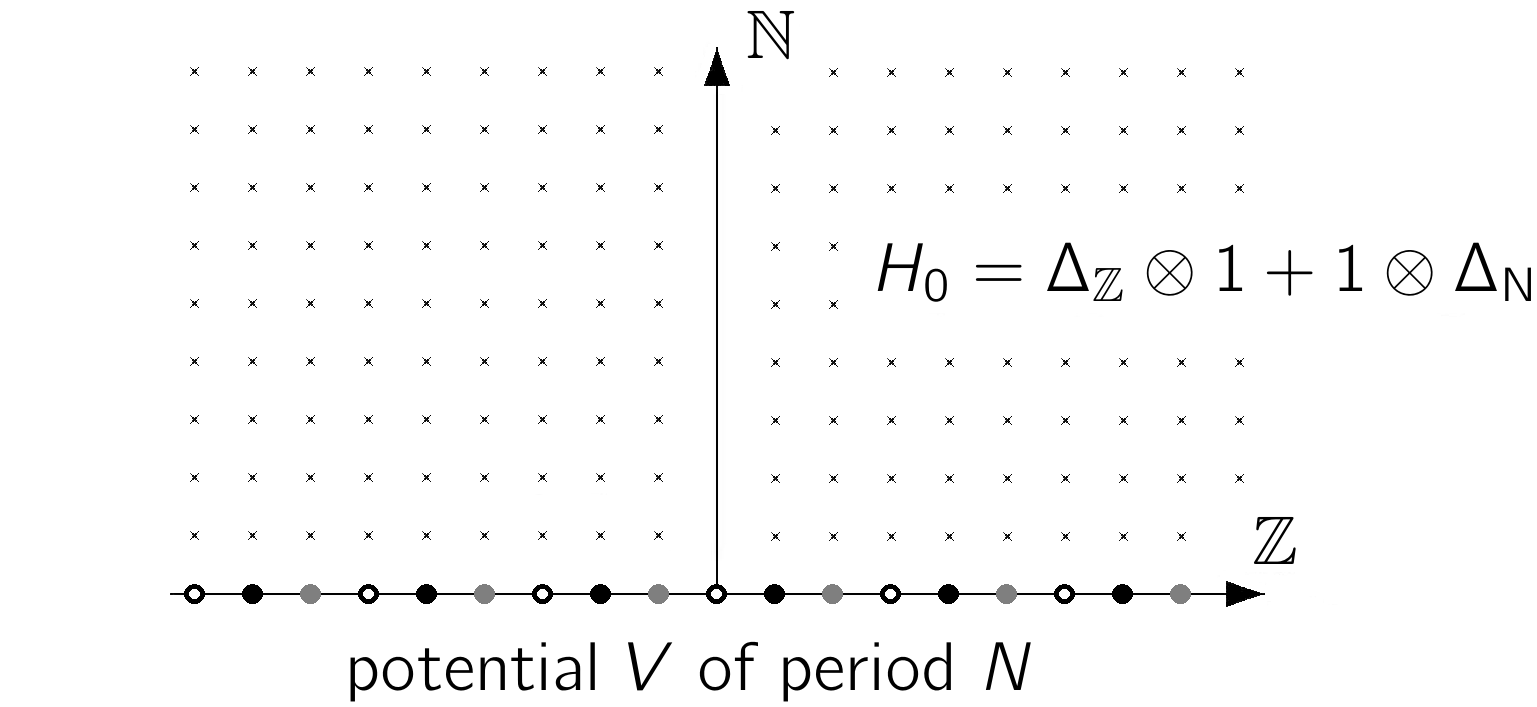}
\caption{Sketch of the two-dimensional discrete model}
\end{figure}

Since the operators $H_0$ and $H$ are $N$-periodic in the $x$-variable, they can be
decomposed using a Bloch-Floquet transformation. Namely, if we set
$\h:=\ltwo\big([0,\pi),\tfrac{\d\omega}\pi;\C^N\big)$ and
$\HH:=\int_{[0,2\pi]}^\oplus\h\;\!\tfrac{\d\theta}{2\pi}$, then it can be shown that
$H_0$ and $H$ are unitarily equivalent to the direct integral operators in $\HH$
\begin{equation}\label{eq_int_H_0}
\int_{[0,2\pi]}^\oplus H^\theta_0\;\!\tfrac{\d\theta}{2\pi}
\quad\hbox{with}\quad
H^\theta_0:=2\cos(\Omega)+A^\theta
\end{equation}
and
\begin{equation}\label{eq_int_H}
\int_{[0,2\pi]}^\oplus H^\theta\;\!\tfrac{\d\theta}{2\pi}
\quad\hbox{with}\quad
H^\theta:=2\cos(\Omega)+A^\theta+\diag(v)P_0
\end{equation}
where $\cos(\Omega)$ is the multiplication operator by the function
$\omega\to\cos(\omega)$ in $\h$, $A^\theta$ is the $N\times N$ hermitian matrix
\begin{equation}\label{eq_matrix}
A^\theta:=
\left(\begin{smallmatrix}
0 & 1 & 0 &\cdots & 0 &\e^{-i\theta}\\
1 & 0 & 1 &\ddots &  & 0\\
0 & 1 &\ddots &\ddots &\ddots &\vdots\\
\vdots &\ddots &\ddots &\ddots & 1 & 0\\
0 & &\ddots & 1 & 0 & 1\\
\e^{i\theta} & 0 &\cdots & 0 & 1 & 0
\end{smallmatrix}\right),
\end{equation}
and
\begin{equation}\label{eq_P_0}
\big(\diag(v)\;\!\f(\theta,\;\!\cdot\;\!)\big)_j
:=v(j)\;\!\f_j(\theta,\;\!\cdot\;\!)
\quad\hbox{and}\quad
\big(P_0\;\!\f(\theta,\;\!\cdot\;\!)\big)_j
:=\int_0^\pi\f_j(\theta,\omega)\tfrac{\d\omega}\pi
\end{equation}
for $\f\in\HH$, $j\in\{1,\dots,N\}$ and a.e. $\theta\in[0,2\pi]$. The main interest of
the above representation is that for each fixed $\theta$ the operator $\diag(v)P_0$ is
a finite rank perturbation of the operator $H_0^\theta$.

\begin{Remark}\label{rem_A_theta}
A direct inspection shows that the matrix $A^\theta$ has eigenvalues
$$
\lambda_j^\theta:=2\cos\left(\tfrac{\theta+2\pi\;\!j}N\right),\quad j\in\{1,\dots,N\},
$$
with corresponding eigenvectors $\xi_j^\theta\in\C^N$ having components
$\big(\xi_j^\theta\big)_k:=\e^{i(\theta+2\pi j)k/N}$, $j,k\in\{1,\dots,N\}$. Using the
notation $\P_j^\theta$ for the orthogonal projection associated to $\xi_j^\theta$, we
thus can write $A^\theta$ as
$A^\theta=\sum_{j=1}^N\lambda_j^\theta\;\!\P_j^\theta$.
\end{Remark}

Due to the unitary equivalences, the analysis of the pair of operators $(H,H_0)$ in
the Hilbert space $\H$ reduces to the analysis of the family of pairs of operators
$(H^\theta,H^\theta_0)$ indexed by the quasi-momentum $\theta\in[0,2\pi]$ in the
Hilbert space $\h$. Therefore, from now on we present our results for the operators
$H^\theta$ and $H^\theta_0$ at fixed $\theta$, and come back to the initial pair
$(H,H_0)$ later on.

Constructing a spectral representation of $H^\theta_0$ is fairly direct. Intuitively,
it amounts to diagonalising the matrix $A^\theta$ and linearising the function $\cos$.
More precisely, first we define for $\theta\in[0,2\pi]$ and $j\in\{1,\dots,N\}$ the
sets
$$
I_j^\theta:=(\lambda_j^\theta-2,\lambda_j^\theta+2)
\quad\hbox{and}\quad
I^\theta:=\cup_{j=1}^NI^\theta_j.
$$
Next, we define the fiber Hilbert spaces
$$
\Hrond^\theta(\lambda)
:=\sp\big\{\P^\theta_j\C^N\mid\hbox{$j\in\{1,\dots,N\}$ such that
$\lambda\in I_j^\theta$}\big\}\subset\C^N,\quad\lambda\in I^\theta,
$$
and the corresponding direct integral Hilbert space
$$
\Hrond^\theta:=\int_{I^\theta}^\oplus\Hrond^\theta(\lambda)\;\!\d\lambda.
$$
Then, it is easily verified that the operator $\F^\theta:\h\to\Hrond^\theta$ defined
by
$$
\big(\F^\theta g\big)(\lambda)
:=\pi^{-1/2}\sum_{\{j\mid\lambda\in I_j^\theta\}}
\big(4-(\lambda-\lambda_j^\theta)^2\big)^{-1/4}\;\!\P^\theta_j
g\Big(\arccos\Big(\tfrac{\lambda-\lambda_j^\theta}2\Big)\Big),
\quad g\in\h,~\hbox{a.e. $\lambda\in I^\theta$,}
$$
is unitary, with adjoint $(\F^\theta)^*:\Hrond^\theta\to\h$ given by
\begin{equation}\label{eq_adjoint}
\big((\F^\theta)^*\zeta\big)(\omega)
:=\big(2\pi\sin(\omega)\big)^{1/2}\sum_{j=1}^N\P_j^\theta\zeta
\big(2\cos(\omega)+\lambda_j^\theta\big),
\quad\zeta\in\Hrond^\theta,~\hbox{a.e. $\omega\in[0,\pi)$.}
\end{equation}
In addition, $\F^\theta$ diagonalises the Hamiltonian $H_0^\theta$. Namely, for all
$\zeta\in\Hrond^\theta$ and a.e. $\lambda\in I^\theta$ one has
$$
\big(\F^\theta H_0^\theta\;\!(\F^{\theta})^*\zeta\big)(\lambda)
=\lambda\;\!\zeta(\lambda)
=\big(X^\theta\zeta\big)(\lambda),
$$
with $X^\theta$ the (bounded) operator of multiplication by the variable in
$\Hrond^\theta$. As a consequence, one infers that $H_0^\theta$ has purely absolutely
continuous spectrum equal to
\begin{equation}\label{eq_spec_H^theta_0}
\textstyle\sigma(H_0^\theta)
=\overline{\Ran(X^\theta)}
=\overline{I^\theta}
=\big[(\min_j\lambda_j^\theta)-2,(\max_j\lambda_j^\theta)+2\big]
\subset[-4,4]
\end{equation}
and also that
$\sigma(H_0)=\overline{\bigcup_{\theta\in[0,2\pi]}\sigma(H_0^\theta)}=[-4,4]$.
Moreover, the spectral representation of $H^\theta_0$ naturally leads to the notion of
thresholds of $H^\theta_0$; namely, the set $\T^\theta$ of real values where the
spectrum of $H^\theta_0$ presents a change of multiplicity:
\begin{equation}\label{eq_thresholds}
\T^\theta:=\big\{\lambda_j^\theta\pm2\mid j\in\{1,\dots,N\}\big\}.
\end{equation}

The next step is the analysis of the operator $H^\theta$, which is detailed in Section
\ref{sec_H_theta}. In short, we determine the spectral properties of $H^\theta$ and we
establish resolvent expansions for $H^\theta$ near the thresholds. Based on the
resolvent expansions, we also derive various properties of the scattering operator for
the pair $(H^\theta,H^\theta_0)$.

Regarding the spectral analysis, the main result is a necessary and sufficient
condition for the existence of eigenvalues of $H^\theta$. To state it, we use standard
notations borrowed from \cite{JN01,Yaf92}. First of all, we decompose the matrix
$\diag(v):=(v(1),\dots,v(N))$ as the product $\diag(v)=\u\v^2$, where
$\v:=|\diag(v)|^{1/2}$ and $\u:=\sgn\big(\diag(v)\big)$ is the diagonal matrix with
components
$$
\u_{jj}
=\sgn\big(\diag(v)\big)_{jj}=
\begin{cases}
+1 &\hbox{if $v(j)\ge0$}\\
-1 &\hbox{if $v(j)<0$,}
\end{cases}
\quad j\in\{1,\dots,N\}.
$$
We also introduce the functions
\begin{equation}\label{def_beta}
\beta_j^\theta(z):=\big|(z-\lambda_j^\theta)^2-4\big|^{1/4},
\quad j\in\{1,\dots,N\},~\theta\in[0,2\pi],~z\in\C.
\end{equation}
The spectral result for $H^\theta$ then reads as follows (recall Remark
\ref{rem_A_theta} for the definitions of $\lambda_j^\theta$ and $\P_j^\theta$):

\begin{Proposition}\label{proposition_kernel}
A value $\lambda\in\R\setminus\T^\theta$ is an eigenvalue of $H^\theta$ if and only if
$$
{\mathcal K}
:=\ker\left(\u+\sum_{\{j\mid\lambda<\lambda_j^\theta-2\}}\tfrac{\v\;\!\P_j^\theta\v}
{\beta_j^\theta(\lambda)^2}-\sum_{\{j\mid\lambda>\lambda_j^\theta+2\}}
\tfrac{\v\;\!\P_j^\theta\v}{\beta_j^\theta(\lambda)^2}\right)
\bigcap\left(\cap_{\{j\mid\lambda\in I^\theta_j\}}\ker\big(\P_j^\theta\v\big)\right)
\ne\{0\},
$$
in which case the multiplicity of $\lambda$ equals the dimension of $\mathcal K$.
\end{Proposition}

Next, we use a general approach for resolvent expansions \cite{JN01,RT16} to derive
detailed asymptotic resolvent expansions for $H^\theta$. For that purpose, we introduce the
bounded operator $G:\h\to\C^N$ given by
\begin{equation}\label{def_G}
(Gg)_j:=\v_{jj}\int_0^\pi g_j(\omega)\;\!\tfrac{\d\omega}\pi
=|v(j)|^{1/2}\int_0^\pi g_j(\omega)\;\!\tfrac{\d\omega}\pi,
\quad g\in\h,~j\in\{1,\dots,N\}.
\end{equation}
Then, the results we obtain about the resolvent of $H^\theta$ are formulated as
asymptotic expansions for the operator
\begin{equation}\label{eq_M_intro}
M^\theta(\lambda+i\varepsilon)
:=\big(\u+G(H_0^\theta-\lambda-i \varepsilon)^{-1}G^*\big)^{-1},
\quad\lambda,\varepsilon\in\R,~\varepsilon\ne0,
\end{equation}
as $\varepsilon\to0$.
They are expressed in terms of projections $S_0,S_1,S_2$ in $\C^N$ of decreasing
range, with the most singular divergences of the expansions taking place in the ranges
of the projections of higher indices (the greater the divergence, the smaller the
subspace where it takes place, see Proposition \ref{Prop_asymp} for details). The
asymptotic expansions are valid for any point $\lambda$ in the spectrum of $H^\theta$.
That is, when $\lambda$ is a threshold of $H^\theta$, when $\lambda$ is an eigenvalue
of $H^\theta$, and when $\lambda$ is neither a threshold, nor an eigenvalue of
$H^\theta$. The expansions imply as a by-product the finiteness of point spectrum of
$H^\theta$, see Remark \ref{remark_no_accumu}.

Once obtained the asymptotic expansions for the operator \eqref{eq_M_intro}, we can
establish our next main result, which is the the continuity of the scattering matrix.
The previous works on the two-dimensional discrete model studied here have discussed
the problem of the existence and the completeness of the wave operators. Such results
lead to the existence and the unitarity of the scattering operator, but do not say
anything about the continuity of the scattering matrix. And this continuity property
will be needed in our second paper in order to apply the $C^*$-algebraic technics
leading to the index theorem.

To formulate the continuity property of the scattering matrix, we first note that the
wave operators
$$
W_{\pm}^\theta:=\slim_{t\to\pm\infty}\e^{itH^\theta}\e^{-itH^\theta_0}
$$
exist and are complete since the difference $H^\theta-H^\theta_0$ is a finite rank
operator, see \cite[Thm.~X.4.4]{Kat95}. As a consequence, the scattering operator
$
S^\theta:=(W_+^\theta)^*W_-^\theta
$
is a unitary operator in $\h$ commuting with $H^\theta_0$. Therefore, $S^\theta$ is
decomposable in the spectral representation of $H^\theta_0$, that is,
$$
\big(\F^\theta S^\theta (\F^\theta)^*h\big)(\lambda)=S^\theta(\lambda)h(\lambda),
\quad\hbox{$h\in\Hrond^\theta$, a.e. $\lambda\in\sigma(H^\theta_0)$,}
$$
where $S^\theta(\lambda)$ (the scattering matrix at energy $\lambda$) is a unitary
operator in $\Hrond^\theta(\lambda)$. To give an explicit formula for
$S^\theta(\lambda)$, Proposition \ref{proposition_kernel} and the asymptotic
expansions play a key role. Indeed, it follows from them that the limit
$$
M^\theta(\lambda+i0):=\lim_{\varepsilon\searrow0}
\big(\u+G(H_0^\theta-\lambda-i\varepsilon)^{-1}G^*\big)^{-1}
$$
exists and belongs to $\B(\C^N)$ for each
$
\lambda\in\sigma(H_0^\theta)\setminus\big(\T^\theta\cup\sigma_{\rm p}(H^\theta)\big)
$,
where $\sigma_{\rm p}(H^\theta)$ is the point spectrum of $H^\theta$. And then, a
computation using stationary formulas \cite[Sec.~2.8]{Yaf92} shows for
$
\lambda\in(I^\theta_j\cap I^\theta_{j'})
\setminus\big(\T^\theta\cup\sigma_{\rm p}(H^\theta)\big)
$
and $j,j'\in\{1,\dots,N\}$ that the channel scattering matrix
$S^\theta(\lambda)_{jj'}:=\P^\theta_jS^\theta(\lambda)\P^\theta_{j'}$ is given
by\footnote{By \emph{channel} we simply mean a choice of indices $(j,j')$, but observe
that $S^\theta(\lambda)_{jj'}$ describes the possible transition from the band
$I^\theta_{j'}$ to the band $I^\theta_{j}$.}
\begin{equation}\label{eq_S_lambda}
S^\theta(\lambda)_{jj'}
=\delta_{jj'}-2i\;\!\beta_j^\theta(\lambda)^{-1}\P_j^\theta\v\;\!
M^\theta(\lambda+i0)\v\;\!\P_{j'}^\theta\;\!\beta_{j'}^\theta(\lambda)^{-1},
\end{equation}
where the operator
$\delta_{jj'}\in\B\big(\P_{j'}^\theta\;\!\C^N;\P_j^\theta\;\!\C^N\big)$ is defined by
$\delta_{jj'}:=1$ if $j=j'$ and $\delta_{jj'}:=0$ otherwise. Now, an explicit formula
for $G(H_0^\theta-\lambda-i0)^{-1}G^*$ (see \eqref{eq_sandwich}) implies the
continuity of the map
$$
(I^\theta_j\cap I^\theta_{j'})
\setminus\big(\T^\theta\cup\sigma_{\rm p}(H^\theta)\big)
\ni\lambda\mapsto S^\theta(\lambda)_{jj'}
\in\B\big(\P_{j'}^\theta\C^N;\P_j^\theta\C^N\big).
$$
Therefore, in order to completely establish the continuity of the channel scattering
matrices $S^\theta(\lambda)_{jj'}$, what remains is to describe the behaviour of
$S^\theta(\lambda)_{jj'}$ as
$\lambda\to\lambda_\star\in\T^\theta\cup\sigma_{\rm p}(H^\theta)$. We consider
separately the behaviour of $S^\theta(\lambda)_{jj'}$ at thresholds and at embedded
eigenvalues, starting with the thresholds. For that purpose, we first observe that for
each $\lambda\in\T^\theta$, a channel can already be open at the energy $\lambda$ (in
which case one has to show the existence and the equality of the limits from the right
and from the left), it can open at the energy $\lambda$ (in which case one only has to
show the existence of the limit from the right), or it can close at the energy
$\lambda$ (in which case one only has to show the existence of the limit from the
left). Therefore, we fix $\lambda\in\T^\theta$, and consider the matrix
$S^\theta(\lambda+\varepsilon)_{jj'}$ for suitable $\varepsilon\in\R$. In this
setting, our continuity result can be stated as (see Theorem \ref{thm_cont} for
a more general statement including the values of the limits):

\begin{Theorem}
Let $\lambda\in\T^\theta$, take $\varepsilon\in \R$  with $|\varepsilon|$ small
enough, and let $j,j'\in\{1,\dots,N\}$.
\begin{enumerate}
\item[(a)] If $\lambda\in I^\theta_j\cap I^\theta_{j'}$, then the limit
$\lim_{\varepsilon\to0}S^\theta(\lambda+\varepsilon)_{jj'}$ exists.
\item[(b)] If $\lambda\in\overline{I^\theta_j}\cap\overline{I^\theta_{j'}}$ and
$\lambda+\varepsilon\in I^\theta_j\cap I^\theta_{j'}$ for $\varepsilon>0$ small
enough, then the limit
$\lim_{\varepsilon\searrow0}S^\theta(\lambda+ \varepsilon)_{jj'}$ exists.
\item[(c)] If $\lambda\in\overline{I^\theta_j}\cap\overline{I^\theta_{j'}}$ and
$\lambda-\varepsilon\in I^\theta_j\cap I^\theta_{j'}$ for $\varepsilon>0$ small
enough, then the limit
$\lim_{\varepsilon\searrow0}S^\theta(\lambda-\varepsilon)_{jj'}$ exists.
\end{enumerate}
\end{Theorem}

We also establish the continuity of the scattering matrix at embedded eigenvalues not
located at thresholds, see Theorem \ref{thm_cont_bis} for more details:

\begin{Theorem}
Let $\lambda\in\sigma_{\rm p}(H^\theta)\setminus\T^\theta$, take $\varepsilon\in\R$
with $|\varepsilon|>0$ small enough, and let $j,j'\in\{1,\dots,N\}$. Then, if
$\lambda\in I^\theta_j\cap I^\theta_{j'}$, the limit
$\lim_{\varepsilon\to0}S^\theta(\lambda+\varepsilon)_{jj'}$ exists.
\end{Theorem}

Our final results concern the wave operator $W_-^\theta$. By using the spectral
representation of $H_0^\theta$ and a stationary representation formula for
$W_-^\theta$, we can express $W_-^\theta$ as the sum of two terms. Namely, we obtain
for suitable $\xi,\zeta\in \Hrond^\theta$
\begin{align}
&\big\langle\F^\theta(W_-^\theta-1)(\F^{\theta })^*\xi,
\zeta\big\rangle_{\Hrond^\theta}\nonumber\\
&=-\pi^{-1/2}\sum_{j=1}^N\int_{I^\theta_j}\lim_{\varepsilon\searrow0}
\left\langle\v M^\theta(\lambda+i\varepsilon)\v\;\!\gamma_0(\F^{\theta})^*
\delta_\varepsilon(X^\theta-\lambda)\xi,
\int_{I_j^\theta}\tfrac{\beta_j^\theta(\mu)^{-1}}{\mu-\lambda+i\varepsilon}\;\!
\zeta_j(\mu)\;\!\d\mu\right\rangle_{\C^N}\d\lambda\label{eq_leading_1}\\
&\quad-\pi^{-1/2}\sum_{j=1}^N\int\limits_{\sigma(H^\theta_0)\setminus I^\theta_j}
\lim_{\varepsilon\searrow0}\left\langle\v M^\theta(\lambda+i\varepsilon)\v\;\!\gamma_0
(\F^{\theta})^*\delta_\varepsilon(X^\theta-\lambda)\xi,
\int_{I_j^\theta}\tfrac{\beta_j^\theta(\mu)^{-1}}{\mu-\lambda+i\varepsilon}\;\!
\zeta_j(\mu)\;\!\d\mu\right\rangle_{\C^N}\d\lambda.\label{eq_remainder_1}
\end{align}
with $\delta_\varepsilon(X^\theta-\lambda)
:=\frac{\pi^{-1}\varepsilon}{(X^\theta-\lambda)^2+\varepsilon^2}$ and
$\gamma_0:\h\to\C^N$ given by
\begin{equation}\label{def_gamma0}
(\gamma_0g)_j:=\int_0^\pi g_j(\omega)\tfrac{\d\omega}\pi,
\quad g\in\h,~j\in\{1,\dots,N\}.
\end{equation}
The main term \eqref{eq_leading_1} could be interpreted as an on-shell contribution,
while the remainder term \eqref{eq_remainder_1} could be interpreted as an off-shell
contribution.

The interest of such a decomposition is that the main term is equal to the product of
an explicit operator independent of the potential, and the operator $S^\theta-1$.
Namely, we show in Section \ref{sec_main} that, up to a compact term, the operator
corresponding to \eqref{eq_leading_1} is unitarily equivalent to
\begin{equation}\label{eq_intro}
\tfrac12\big(1-\tanh(\pi\Dt)-i\cosh(\pi\Dt)^{-1}\tanh(\Xt)\big)(S^\theta-1),
\end{equation}
where $\Xt$ and $\Dt$ are representations of the canonical position and momentum
operators in the Hilbert space $\h$. This formula is obtained by considering the
on-shell contribution in a rescaled energy representation whose importance has been
revealed in \cite{BSB,SB16} and which was also used explicitly in \cite{IT19} and
implicitly in \cite{KR08,PR}.

The analysis of the remainder term \eqref{eq_remainder_1} is more involved, and
depends on the value of $\theta$. When $\theta\ne0$ or when $\theta=0$ and $N$ is odd,
the operator corresponding to \eqref{eq_remainder_1} extends continuously to a compact
operator. Since the operator \eqref{eq_intro} is never compact, this shows that the
remainder term can indeed be considered small compared to the leading term. On the
other hand, when $\theta=0$ and $N$ is even, more analysis is required. In this case,
the compacity argument does not work when the energy bands $[-4,0]$ and $[0,4]$ of
$H^0$ exactly touch, without overlaping, see Remark \ref{remark_bad}. However, if the
vectors $\v\xi_N^0$ and $\v\xi^0_{N/2}$ are linearly independent (see Remark
\ref{rem_A_theta} when $\theta=0$ for the definition of the vectors), one can still
show that the remainder term is compact. We thus call \emph{degenerate case} the very
exceptional case where $\theta=0$, $N$ is even, and $\v\xi_N^0$ and $\v\xi^0_{N/2}$
are linearly dependent. A direct inspection shows that it takes place if and only if
the matrix $\v$ has the very particular form
\begin{equation}\label{eq_spf}
\v=
\left(\begin{smallmatrix}
v(1)&&&&0\\
&0&&&\\
&&v(3)&&\\
&&&0&\\
0&&&&\ddots\\
\end{smallmatrix}\right)
\quad\hbox{or}\quad
\v=
\left(\begin{smallmatrix}
0&&&&0\\
&v(2)&&&\\
&&0&&\\
&&&v(4)&\\
0&&&&\ddots\\
\end{smallmatrix}\right).
\end{equation}
In the degenerate case, the remainder term is bounded but not compact. However, in our
second paper, we will show that the remainder term can still be considered small
compared to the leading term, once suitable $C^*$-algebras are introduced.

Summing up what precedes, we get (see Theorem \ref{thm_formula_theta} for more details):

\begin{Theorem}\label{thm_theta_intro}
For any $\theta\in[0,2\pi]$, one has the equality
$$
W_-^\theta-1
=\tfrac12\big(1-\tanh(\pi\Dt)-i\cosh(\pi\Dt)^{-1}\tanh(\Xt)\big)(S^\theta-1)
+\Kt^\theta,
$$
with $\Kt^\theta$ compact in the nondegenerate cases, and $\Kt^0$ bounded in the
degenerate case.
\end{Theorem}

This type of formula has been obtained for various models having a finite point
spectrum: first in \cite{KR06,KR07}, and then in various other papers summarised in
the review article \cite{Ric16}. Similar formulas have also been independently
obtained in \cite{BSB} and in \cite{IT19}, and some generalisations to the case of an
infinite number of eigenvalues can be found in \cite{IR1,IR2}.

The final step is to combine the formulas for the wave operators for all quasi-momenta
$\theta$ to obtain a new representation formula for the wave operators
$W_\pm:=\slim_{t\to\pm\infty}\e^{itH}\e^{-itH_0}$ of the initial pair $(H,H_0)$. For
this, we first note from the direct integral decompositions
\eqref{eq_int_H_0}-\eqref{eq_int_H}, from the existence and completeness of
$W_\pm^\theta$ for each $\theta\in[0,2\pi]$, and from \cite[Sec.~2.4]{Fra03}, that
$W_\pm$ exist and have same range. In addition, the wave operators $W_\pm$ and the
scattering operator $S:=(W_+)^*W_-$ are unitarily equivalent to the direct integral
operators in $\HH$
$$
\int_{[0,2\pi]}^\oplus W_\pm^\theta\;\!\tfrac{\d\theta}{2\pi}
\quad\hbox{and}\quad
\int_{[0,2\pi]}^\oplus S^\theta\;\!\tfrac{\d\theta}{2\pi}.
$$
Therefore, by collecting the formulas obtained in Theorem \ref{thm_theta_intro} for
$W_-^\theta-1$ in each fiber Hilbert space $\h$, we obtain a new formula for $W_--1$
(and thus also for $W_+$ if we use the relation $W_+=W_-S^*$):

\begin{Theorem}\label{thm_the_goal}
The operator $W_--1$ is unitarily equivalent to the direct integral operator in $\HH$
\begin{equation}\label{eq_the_goal}
\int_{[0,2\pi]}^\oplus\left(\tfrac12\big(1-\tanh(\pi\Dt)
-i\cosh(\pi\Dt)^{-1}\tanh(\Xt)\big)(S^\theta-1)+\Kt^\theta\right)
\tfrac{\d\theta}{2\pi},
\end{equation}
with $\Kt^\theta$ as in Theorem \ref{thm_theta_intro}.
\end{Theorem}

A more detailed version of this result is presented in Theorem \ref{thm_formula}, with
the unitary equivalence explicitly stated. Now, even though this theorem is the
culminating result of this paper, it is also the starting point for future
investigations. Indeed, in recent years, similar formulas for the wave operators have
been at the root of topological index theorems in scattering theory generalising the
so-called Levinson's theorem. To some extent, these index theorems encode the fact
that the wave operators are partial isometries which relate, through the projection on
their cokernels, the scattering states of a system to its bound states. In our
situation, it can be shown using the direct integral representation
\eqref{eq_the_goal} that the states which belong to the cokernel of $W_-$ are no more
bound states but surface states. Therefore, the theorem mentioned at the beginning of
this introduction will be an index theorem about surface states based on Theorem
\ref{thm_the_goal}. In fact, a result of this type (a relation between the total
density of surface states and the density of the total time delay) has already
appeared in \cite{SB16}. Let us also mention the work \cite{GP} which contains a
bulk-edge correspondence for two-dimensional topological insulators, and whose proof
is partially based on scattering theory. For our model, the necessary $C^*$-algebraic
framework will be introduced in a second paper, and the continuity of the scattering
matrix and the existence of its limits at thresholds established here will play a
crucial role for the choice of the $C^*$-algebras. The $\theta$-dependence of all the
operators will also be a key ingredient for the construction. More information on
these issues, and the applications of the analytical results obtained here, will be
presented in the second paper.

\section{Direct integral decompositions of $H_0$ and $H$}\label{sec_direct}
\setcounter{equation}{0}

Before describing the direct integral decompositions of $H_0$ and $H$, let us observe
that the discrete Neumann operator $\Delta_\NN$ has been chosen so that it is a
natural restriction of the discrete adjacency operator $\Delta_\Z$ in $\ell^2(\Z)$.
Indeed, if we decompose $\ell^2(\Z)$ as a direct sum of even and odd functions
$\ell^2(\Z)=\ell^2_{\rm even}(\Z)\oplus\ell^2_{\rm odd}(\Z)$, then $\Delta_\Z$ is
reduced by this decomposition and satisfies the equality $\Delta_\NN:=\U\Delta_\Z\U^*$
with $\U:\ell^2_{\rm even}(\Z)\to\ell^2(\N)$ the unitary operator given by
$$
(\U\varphi)(n):=
\begin{cases}
\varphi(0) &\hbox{if $n=0$}\\
2^{1/2}\;\!\varphi(n) &\hbox{if $n\ge1$.}
\end{cases}
$$

Since $H_0$ and $H$ are periodic in the $x$-variable, it is natural to decompose them
using a Bloch-Floquet transformation. We sketch this decomposition by providing the
necessary formulas, and leave the details to the reader. For
$\psi\in\H_{\rm fin}:=\{\psi\in\H\mid\hbox{$\supp(\psi)$ is finite}\}$ the
Bloch-Floquet transformation is defined by
$$
\big(\G_1\psi\big)_j(\theta,n):=\sum_{k\in\Z}\e^{-ik\theta}\psi(kN+j,n),
\quad\psi\in\H_{\rm fin},~j\in\{1,\dots,N\},~\theta\in[0,2\pi],~n\in\N,
$$
and $\G_1$ extends to a unitary operator from $\H$ to
$\int_{[0,2\pi]}^\oplus\ell^2\big(\N;\C^N\big)\tfrac{\d\theta}{2\pi}$ satisfying the
relation
$$
\G_1H_0\;\!\G_1^*=\int_{[0,2\pi]}^\oplus H_0(\theta)\;\!\tfrac{\d\theta}{2\pi},
$$
with $H_0(\theta):=\Delta_\NN+A^\theta$ and $A^\theta$ the $N\times N$ hermitian
matrix \eqref{eq_matrix}. Then, we define a second unitary operator
$\G_2:\int_{[0,2\pi]}^\oplus\ell^2\big(\N;\C^N\big)\;\!\tfrac{\d\theta}{2\pi}\to\HH$
(acting on the $n$-variable) given for suitable elements
$f\in\int_{[0,2\pi]}^\oplus\ell^2\big(\N;\C^N\big)\;\!\tfrac{\d\theta}{2\pi}$ by
$$
\big(\G_2f\big)_j(\theta,\omega)
:=2^{1/2}\sum_{n\ge1}\cos(n\;\!\omega)f_j(\theta,n)+f_j(\theta,0)
\quad\hbox{$j\in\{1,\dots,N\}$, a.e. $\theta\in[0,2\pi]$, $\omega\in[0,\pi)$}.
$$
Finally, using the composed unitary operator $\G:=\G_2\;\!\G_1$, we obtain that
$\G H_0\G^*=\int_{[0,2\pi]}^\oplus H^\theta_0\;\!\tfrac{\d\theta}{2\pi}$ with
$H^\theta_0$ as in \eqref{eq_int_H_0}.

Now, to decompose $H=H_0+V$, one only needs to compute the image of $V$ through $\G$.
A straightforward calculation gives
$\G V\G^*=\int_{[0,2\pi]}^\oplus\diag(v)P_0\;\!\tfrac{\d\theta}{2\pi}$ with $\diag(v)$
and $P_0$ as in \eqref{eq_P_0}. Putting what precedes together, we thus obtain that
$\G H\G^*=\int_{[0,2\pi]}^\oplus H^\theta\;\!\tfrac{\d\theta}{2\pi}$ with $H^\theta$
as in \eqref{eq_int_H}. The main interest of the above representation is that for each
fixed $\theta$ the operator $\diag(v)P_0$ is a finite rank perturbation of the
operator $H_0^\theta$. Indeed, if $\{e_j\}_{j=1}^N$ denotes the canonical basis of
$\C^N$, then one has for any $g\in\h$
$$
\diag(v)P_0\;\!g=\sum_{j=1}^Nv(j)(P_0g)_j\;\!e_j,
$$
with the r.h.s. independent of the variable $\omega$.

\begin{Remark}\label{rem_A_theta}
(a) Since $A^0=A^{2\pi}$, the matrices $A^0$ and $A^{2\pi}$ have the same eigenvalues,
that is, for each $j\in\{1,\dots,N\}$ there exists $j'\in\{1,\dots,N\}$ (in general
distinct from $j$) such that $\lambda_j^0=\lambda_{j'}^{2\pi}$.

(b) The eigenvalues of $A^\theta$ have multiplicity $1$, except in the cases
$\theta=0$, $\pi$, and $2\pi$, when they can have multiplicity $2$. Namely, one has:
\begin{enumerate}
\item[(i)] $\lambda_j^0=\lambda_{N-j}^0$ if $N\ge3$ and $j\in\{1,\dots,N-1\}$,
\item[(ii)] $\lambda_N^\pi=\lambda_{N-1}^\pi$ if $N\ge2$, and
$\lambda_j^\pi=\lambda_{N-j-1}^\pi$ if $N\ge4$ and $j\in\{1,\dots,N-2\}$,
\item[(iii)] $\lambda_{N-2}^{2\pi}=\lambda_{N}^{2\pi}$ if $N\ge3$, and
$\lambda_j^{2\pi}=\lambda_{N-j-2}^{2\pi}$ if $N\ge5$ and $j\in\{1,\dots,N-3\}$.
\end{enumerate}
\end{Remark}

To conclude the section, we provide as an example the explicit formulas for the matrix
$A^\theta$ and its eigenvalues $\lambda_j^\theta$, eigenvectors $\xi_j^\theta$, and
orthogonal projections $\P_j^\theta$ in the case $N=2:$

\begin{Example}[Case $N=2$]\label{ex_period_2}
When the potential $v$ has period $N=2$, the matrix $A^\theta$ takes the form
$$
A^\theta
=\begin{pmatrix}
0 & 1+\e^{-i\theta}\\
1+\e^{i\theta} & 0
\end{pmatrix}.
$$
It has eigenvalues $\lambda_1^\theta=-2\cos(\theta/2)$ and
$\lambda_2^\theta=2\cos(\theta/2)$, eigenvectors
$
\xi_1^\theta
=\left(\begin{smallmatrix}
-\e^{i\theta/2}\\\e^{i\theta}
\end{smallmatrix}\right)
$
and
$
\xi_2^\theta=
\left(\begin{smallmatrix}
\e^{i\theta/2}\\\e^{i\theta}
\end{smallmatrix}\right),
$
and orthogonal projections
$$
\P_1^\theta
=\tfrac12
\begin{pmatrix}
1 & -\e^{-i\theta/2}\\
-\e^{i\theta/2} & 1
\end{pmatrix}
\quad\hbox{and}\quad
\P_2^\theta
=\tfrac12
\begin{pmatrix}
1 &\e^{-i\theta/2}\\
\e^{i\theta/2} & 1
\end{pmatrix}.
$$
\end{Example}

\section{Analysis of the fibered Hamiltonians $H^\theta$}\label{sec_H_theta}
\setcounter{equation}{0}

In this section, we establish spectral properties and derive asymptotic resolvent
expansions for the fibered Hamiltonians $H^\theta$. Using the resolvent expansions, we
also determine properties the scattering operator $S^\theta$ for the pair
$(H^\theta,H^\theta_0)$. Some of the proofs are differed to the Appendix in order to
make the reading more pleasant.

\subsection{Spectral analysis of the fibered Hamiltonians $H^\theta$}

Recall first that the $N\times N$ matrices $\u$ and $\v$ have been introduced before
\eqref{def_beta}, and that the operator $G:\h\to\C^N$ has been defined in
\eqref{def_G}. The adjoint $G^*:\C^N\to\h$ is then given by
$$
(G^*\xi)_j(\omega):=|v(j)|^{1/2}\xi_j,
\quad\xi\in\C^N,~j\in\{1,\dots,N\},~\omega\in[0,\pi).
$$
By setting $R_0^\theta(z):=(H_0^\theta-z)^{-1}$ and $R^\theta(z):=(H^\theta-z)^{-1}$
for $z\in\C\setminus\R$, the resolvent equation takes the form
\begin{equation}\label{eq_resolv_1}
R^\theta(z)
=R_0^\theta(z)-R_0^\theta(z)G^*\big(\u+GR_0^\theta(z)G^*\big)^{-1}GR_0^\theta(z)
\end{equation}
or the equivalent form
\begin{equation}\label{eq_resolv_2}
GR^\theta(z)G^*=\u-\u\big(\u+GR_0^\theta(z)G^*\big)^{-1}\u.
\end{equation}
These equations are rather standard and can be deduced from the usual resolvent
equations, see for instance \cite[Sec.~1.9]{Yaf92}.

Motivated by the formula \eqref{eq_resolv_2}, we now analyse the operator
$(\u+GR_0^\theta(z)G^*)^{-1}$ which belongs to the set $\B(\C^N)$ of $N\times N$
complex matrices for each $z\in\C\setminus\R$. In the lemma below, we prove the
existence of the limit
$$
GR_0^\theta(\lambda+i\;\!0)G^*
:=\lim_{\varepsilon\searrow0}GR_0^\theta(\lambda+i\varepsilon)G^*
$$
for appropriate values of $\lambda\in\R$, and we provide an expression for the limit.
For the proof, we use the convention that the square root $\sqrt z$ of a complex
number $z\in\C\setminus[0,\infty)$ is chosen so that $\im(\sqrt z)>0$. We also define
the unit circle $\S^1:=\{z\in\C\mid|z|=1\}$, the unit disc $\D:=\{z\in\C\mid|z|<1\}$,
and recall that  the set of thresholds $\T^\theta$ and the functions $\beta_j^\theta$
have been introduced in  \eqref{eq_thresholds} and \eqref{def_beta}, respectively.
Note that in general, the set $\T^\theta$ contains $2N$ elements, but it may contain
fewer elements if some eigenvalues of $A^\theta$ have multiplicity $2$, see Remark
\ref{rem_A_theta}(b).

\begin{Lemma}\label{lemma_sandwich}
For each $\theta\in[0,2\pi]$ and $\lambda\in\R\setminus\T^\theta$ the following
equality holds in $\B(\C^N):$
\begin{equation}\label{eq_sandwich}
GR_0^\theta(\lambda+i\;\!0)G^*
=\sum_{\{j\mid\lambda<\lambda_j^\theta-2\}}\tfrac{\v\;\!\P_j^\theta\v}
{\beta_j^\theta(\lambda)^2}
+i\sum_{\{j\mid\lambda\in I^\theta_j\}}\tfrac{\v\;\!\P_j^\theta\v}
{\beta_j^\theta(\lambda)^2}
-\sum_{\{j\mid\lambda>\lambda_j^\theta+2\}}\tfrac{\v\;\!\P_j^\theta\v}
{\beta_j^\theta(\lambda)^2}.
\end{equation}
\end{Lemma}

The proof of Lemma \ref{lemma_sandwich}, which essentially consists in a careful
application of the residue theorem, is given in the Appendix. Based on this lemma, the
characterisation of the point spectrum of the operator $H^\theta$ stated in
Proposition \ref{proposition_kernel} can be obtained. The proof of the proposition is
similar to that of \cite[Lemma~3.1]{RT17}:

\begin{proof}[Proof of Proposition \ref{proposition_kernel}]
The proof consists in applying \cite[Lemma~4.7.8]{Yaf92}. Once the assumptions of this
lemma are checked, it implies that the multiplicity of an eigenvalue
$\lambda\in\sigma_{\rm p}(H^\theta)\setminus\T^\theta$ equals the multiplicity of the
eigenvalue $1$ of the operator $-GR_0^\theta(\lambda+i\;\!0)G^*\u$. But since
$\u^2=1$, one deduces from Lemma \ref{lemma_sandwich} that
\begin{align}
&-GR_0^\theta(\lambda+i\;\!0)G^*\u\;\!\xi=\xi\quad(\xi\in\C^N)\nonumber\\
&\Leftrightarrow\u\;\!\xi\in
\ker\left(\u+\sum_{\{j\mid\lambda<\lambda_j^\theta-2\}}\tfrac{\v\;\!\P_j^\theta\v}
{\beta_j^\theta(\lambda)^2}
+i\sum_{\{j\mid\lambda\in I^\theta_j\}}\tfrac{\v\;\!\P_j^\theta\v}
{\beta_j^\theta(\lambda)^2}
-\sum_{\{j\mid\lambda>\lambda_j^\theta+2\}}\tfrac{\v\;\!\P_j^\theta\v}
{\beta_j^\theta(\lambda)^2}\right).\label{eq_ker_T_0}
\end{align}
By separating the real and imaginary parts of the operator on the r.h.s. and by noting
that the imaginary part consists in a sum of positive operators, one infers that
\eqref{eq_ker_T_0} reduces to the inclusion $\u\;\!\xi\in\mathcal K$. And since $\u$
is unitary, this implies the assertions.

We are thus left with proving that the assumptions of \cite[Lemma~4.7.8]{Yaf92} hold
in a neighbourhood of $\lambda\in\sigma_{\rm p}(H^\theta)\setminus\T^\theta$. First,
we recall that $H_0^\theta$ has purely absolutely continuous spectrum and spectral
multiplicity constant in a small neighbourhood of $\lambda$, as a consequence of the
spectral representation of $H_0^\theta$.
So, what remains is to prove that the operators $G$ and $\u\;\!G$ are strongly
$H^\theta_0$-smooth with some exponent $\alpha>1/2$ on any closed interval of
$\R\setminus\T^\theta$ (see \cite[Def.~4.4.5]{Yaf92} for the definition of strongly
smooth operators). In our setting, one can check that the $H_0^\theta$-smoothness with
exponent $\alpha>1/2$ coincides with the H\"older continuity with exponent
$\alpha>1/2$ of the functions
$$
\R\setminus\T^\theta\ni\lambda\mapsto\big(\F^\theta G^*\xi\big)(\lambda)
=\pi^{-1/2}\sum_{\{j\mid\lambda\in I_j^\theta\}}\beta^\theta_j(\lambda)^{-1}
\P^\theta_j\v\;\!\xi\in\C^N,\quad\xi\in\C^N.
$$
Since this can be verified directly, as well as when $G^*$ is replaced by $G^*\u$, all
the assumptions of \cite[Lemma~4.7.8]{Yaf92} are satisfied.
\end{proof}

We illustrate the results of Proposition \ref{proposition_kernel} in the case $N=2$,
as we did in Example \ref{ex_period_2}:

\begin{Example}[Case $N=2$, continued]
In the case $N=2$, assume that $v(2x)=0$ and $v(2x+1)=-a^2$ for some $a>0$ and all
$x\in\Z$. Then, $\u=\left(\begin{smallmatrix}-1&0\\0&1\end{smallmatrix}\right)$,
$\v=\left(\begin{smallmatrix}a&0\\0&0\end{smallmatrix}\right)$, and one has for any
$\theta\in[0,\pi]$
$$
\sigma(H^\theta_0)
=\sigma_{\rm ess}(H^\theta)
=\big[-2\cos(\theta/2)-2, 2\cos(\theta/2)+2\big].
$$
Therefore, if $\lambda<-2\cos(\theta/2)-2=\inf\big(\sigma_{\rm ess}(H^\theta)\big)$,
Proposition \ref{proposition_kernel} implies that $\lambda$ is an eigenvalue of
$H^\theta$ if and only if
$$
\ker\left(\u+\tfrac{\v\;\!\P_1^\theta\v}{\beta_1^\theta(\lambda)^2}
+\tfrac{\v\;\!\P_2^\theta\v}{\beta_2^\theta(\lambda)^2}\right)
=\ker\begin{pmatrix}\tfrac{a^2}2\big(\beta_1^\theta(\lambda)^{-2}
+\beta_2^\theta(\lambda)^{-2}\big)-1 & 0
\\ 0 & 1\end{pmatrix}\ne\{0\},
$$
which is verified if and only if
\begin{equation}\label{eq_cond}
\big((\lambda+2\cos(\theta/2))^2-4\big)^{-1/2}
+\big((\lambda-2\cos(\theta/2))^2-4\big)^{-1/2}
=2a^{-2}.
\end{equation}
Now, the l.h.s. is a continuous, strictly increasing function of
$\lambda\in\big(-\infty,-2\cos(\theta/2)-2\big)$ with range equal to $(0,\infty)$. So,
there exists a unique solution to the equation \eqref{eq_cond}. And since a similar
argument holds for $\theta\in[\pi,2\pi]$, we conclude that for each
$\theta\in[0,2\pi]$ the operator $H^\theta$ has an eigenvalue of multiplicity $1$
below its essential spectrum.
\end{Example}

\begin{Example}[Case $N=2$, continued]
Still in the case $N=2$, assume this time that $v(2x)=a^2$ and $v(2x+1)=b^2$ for some
$a,b>0$, $a\ne b$, and all $x\in\Z$. Then, one has for any $\theta\in[0,\pi]$ and
$\lambda>2\cos(\theta/2)+2=\sup\big(\sigma_{\rm ess}(H^\theta)\big)$
$$
\u-\tfrac{\v\;\!\P_1^\theta\v}{\beta_1^\theta(\lambda)^2}
-\tfrac{\v\;\!\P_2^\theta\v}{\beta_2^\theta(\lambda)^2}
=\begin{pmatrix}
1-\tfrac{b^2}2\big(\beta_1^\theta(\lambda)^{-2}+\beta_2^\theta(\lambda)^{-2}\big)
& \tfrac{ab\e^{-i\theta/2}}2\big(\beta_1^\theta(\lambda)^{-2}
-\beta_2^\theta(\lambda)^{-2}\big)\\
\tfrac{ab\e^{i\theta/2}}2\big(\beta_1^\theta(\lambda)^{-2}
-\beta_2^\theta(\lambda)^{-2}\big)
& 1-\tfrac{a^2}2\big(\beta_1^\theta(\lambda)^{-2}+\beta_2^\theta(\lambda)^{-2}\big)
\end{pmatrix},
$$
and the determinant of this matrix is zero if and only if
\begin{equation}\label{eq_det}
\tfrac{a^2+b^2}2\big(\beta_1^\theta(\lambda)^2+\beta_2^\theta(\lambda)^2\big)
-\big(\beta_1^\theta(\lambda)\beta_2^\theta(\lambda)\big)^2-(ab)^2
=0.
\end{equation}
In order to check when this equation has solutions, we set
$\lambda(\mu):=2\cos(\theta/2)+2+\mu$ with $\mu>0$ and
\begin{align*}
f^\theta(\mu)
&:=\beta_1^\theta\big(\lambda(\mu)\big)
=\big(\mu+4\cos(\theta/2)\big)^{1/4}\big(\mu+4+4\cos(\theta/2)\big)^{1/4},\\
g^\theta(\mu)
&:=\beta_2^\theta\big(\lambda(\mu)\big)
=\mu^{1/4}(\mu+4)^{1/4},\\
h^\theta(\mu)
&:=\tfrac{a^2+b^2}2\big(f^\theta(\mu)^2+g^\theta(\mu)^2\big)
-\big(f^\theta(\mu)g^\theta(\mu)\big)^2-(ab)^2.
\end{align*}
With these notations, the equation \eqref{eq_det} reduces to $h^\theta(\mu)=0$. Now,
if $2^{3/2}(a^2+b^2)<(ab)^2$, then
$$
\lim_{\mu\searrow0}h^\theta(\mu)
=2\;\!(a^2+b^2)\cos(\theta/2)^{1/2}\big(1+\cos(\theta/2)\big)^{1/2}-(ab)^2
<2^{3/2}(a^2+b^2)-(ab)^2
<0.
$$
On the other hand, the AM-GM inequality implies that
$$
h^\theta(\mu)
\ge(a^2+b^2)f^\theta(\mu)g^\theta(\mu)
-\big(f^\theta(\mu)g^\theta(\mu)\big)^2-(ab)^2,
$$
with the r.h.s. strictly positive if
$f^\theta(\mu)g^\theta(\mu)\in(\min\{a^2,b^2\},\max\{a^2,b^2\})$. Finally, we have
$\lim_{\mu\to\infty}h^\theta(\mu)=-\infty$. In consequence, if
$2^{3/2}(a^2+b^2)<(ab)^2$, then the function $h^\theta$ is (i) strictly negative for
$\mu$ small enough, (ii) strictly positive on some positive interval, and (iii)
strictly negative for $\mu$ large enough. Since $h^\theta$ is continuous, it follows
that the equation $h^\theta(\mu)=0$ (and thus the equation \eqref{eq_det}) has at
least $2$ distinct solutions for any $\theta\in[0,\pi]$. Since a similar argument
holds for $\theta\in[\pi,2\pi]$, we conclude that for each $\theta\in[0,2\pi]$ the
operator $H^\theta$ has at least $2$ distinct eigenvalues above its essential
spectrum.
\end{Example}

\subsection{Resolvent expansions for the fibered Hamiltonians $H^\theta$}
\label{sec_expansion}

We are now ready to derive resolvent expansions for the fibered Hamiltonians
$H^\theta$ using the inversion formulas and iterative scheme developed in
\cite{RT16,RT17}. For that purpose, we set $\C_+:=\{z\in\C\mid\im(z)>0\}$ and consider
points $z=\lambda-\kappa^2$ with $\lambda\in\R$ and $\kappa$ belonging to the sets
$$
\begin{tabular}{l}
$O(\varepsilon)
:=\big\{\kappa\in\C\mid|\kappa|\in(0,\varepsilon),~\re(\kappa)>0,
~\im(\kappa)<0\big\}$\\
$\widetilde O(\varepsilon)
:=\big\{\kappa\in\C\mid|\kappa|\in(0,\varepsilon),~\re(\kappa)\ge0,
~\im(\kappa)\le0\big\}$
\end{tabular}
\quad(\varepsilon>0).
$$
Note that if $\kappa\in O(\varepsilon)$, then $-\kappa^2\in\C_+$ while if
$\kappa\in\widetilde O(\varepsilon)$, then $-\kappa^2\in\overline{\C_+}$. The main
result of this section then reads as follows:

\begin{Proposition}\label{Prop_asymp}
Let $\lambda\in\T^\theta\cup\sigma_{\rm p}(H^\theta)$, and take
$\kappa\in O(\varepsilon)$ with $\varepsilon>0$ small enough. Then, the operator
$\big(\u+GR_0^\theta(\lambda-\kappa^2)G^*\big)^{-1}$ belongs to $\B(\C^N)$ and is
continuous in the variable $\kappa\in O(\varepsilon)$. Moreover, the continuous
function
$$
O(\varepsilon)\ni\kappa\mapsto
\big(\u+GR_0^\theta(\lambda-\kappa^2)G^*\big)^{-1}\in\B(\C^N)
$$
extends continuously to a function
$
\widetilde O(\varepsilon)\ni\kappa\mapsto\M^\theta(\lambda,\kappa)\in\B(\C^N)
$,
and for each $\kappa\in\widetilde O(\varepsilon)$ the operator
$\M^\theta(\lambda,\kappa)$ admits an asymptotic expansion in $\kappa$. The precise
form of this expansion is given in \eqref{eq_sol_1} and \eqref{eq_sol_2}.
\end{Proposition}

The proof of Proposition \ref{Prop_asymp} relies on an inversion formula which we
reproduce here for convenience; an earlier version of it is also available in
\cite[Prop.~1]{JN01}.

\begin{Proposition}[Proposition 2.1 of \cite{RT16}]\label{prop_inverse}
Let $O\subset\C$ be a subset with $0$ as an accumulation point, and let $\H$ be a
Hilbert space. For each $z\in O$, let $A(z)\in\B(\H)$ satisfy
$$
A(z)=A_0+zA_1(z),
$$
with $A_0\in\B(\H)$ and $\|A_1(z)\|_{\B(\H)}$ uniformly bounded as $z\to0$. Let also
$S\in\B(\H)$ be a projection such that (i) $A_0+S$ is invertible with bounded inverse,
(ii) $S(A_0+S)^{-1}S=S$. Then, for $|z|>0$ small enough the operator $B(z):S\H\to S\H$
defined by
$$
B(z)
:=\tfrac1z\Big(S-S\big(A(z)+S\big)^{-1}S\Big)
\equiv S(A_0+S)^{-1}\Bigg(\sum_{j\ge0}(-z)^j\big(A_1(z)(A_0+S)^{-1}\big)^{j+1}\Bigg)S
$$
is uniformly bounded as $z\to0$. Also, $A(z)$ is invertible in $\H$ with bounded
inverse if and only if $B(z)$ is invertible in $S\H$ with bounded inverse, and in this
case one has
$$
A(z)^{-1}
=\big(A(z)+S\big)^{-1}+\tfrac1z\big(A(z)+S\big)^{-1}SB(z)^{-1}S\big(A(z)+S\big)^{-1}.
$$
\end{Proposition}

Even if the proof of Proposition \ref{Prop_asymp} is a bit lengthy, we prefer to
present it here, in the core of the text, since several notations and auxiliary
operators are introduced in it.

\begin{proof}[Proof of Proposition \ref{Prop_asymp}]
For each $\lambda\in\R$, $\varepsilon>0$ and $\kappa\in O(\varepsilon)$, one has
$\im(\lambda-\kappa^2)>0$. Thus, the operator
$(\u+GR_0^\theta(\lambda-\kappa^2)G^*)^{-1}$ belongs to $\B(\C^N)$ and is continuous
in $\kappa\in O(\varepsilon)$ due to \eqref{eq_resolv_2}. For the other claims, we
distinguish the cases $\lambda\in\T^\theta$ and
$\lambda\in\sigma_{\rm p}(H^\theta)\setminus\T^\theta$, starting with the case
$\lambda\in\T^\theta$. All the operators defined below depend on the choice of
$\lambda$, but for simplicity we do not always write these dependencies.

(i) Assume that $\lambda\in\T^\theta$ and take $\kappa\in O(\varepsilon)$ with
$\varepsilon>0$ small enough. Then, it follows from Lemma \ref{lemma_sandwich} that
$$
GR_0^\theta(\lambda-\kappa^2)G^*
=-\sum_{j\in\{1,\dots,N\}}\tfrac{\v\;\!\P_j^\theta\v}
{\sqrt{(\lambda-\kappa^2-\lambda_j^\theta)^2-4}}
=-\sum_{j\in\{1,\dots,N\}}\tfrac{\v\;\!\P_j^\theta\v}
{\sqrt{(\lambda-\kappa^2-\lambda_j^\theta+2)
(\lambda-\kappa^2-\lambda_j^\theta-2)}}.
$$
Now, let
$
\N_\lambda:=\big\{j\in\{1,\dots,N\}\mid|\lambda-\lambda_j^\theta|=2\big\}
$
and set
$$
\vartheta_j(\kappa)
:=\tfrac1\kappa\sqrt{(\lambda-\kappa^2-\lambda_j^\theta+2)
(\lambda-\kappa^2-\lambda_j^\theta-2)},\quad j\in\N_\lambda.
$$
Then, we have
\begin{equation}\label{eq_vartheta}
\vartheta_j(\kappa)
=\begin{cases}
\sqrt{4+\kappa^2} &\hbox{if $\lambda=\lambda_j^\theta-2$}\\
i\;\!\sqrt{4-\kappa^2} &\hbox{if $\lambda=\lambda_j^\theta+2$}
\end{cases}
\quad\hbox{and}\quad
\lim_{\kappa\to0}\vartheta_j(\kappa)
=\begin{cases}
-2 &\hbox{if $\lambda=\lambda_j^\theta-2$}\\
2i &\hbox{if $\lambda=\lambda_j^\theta+2$.}
\end{cases}
\end{equation}
With these notations, we obtain
$$
\big(\u+GR_0^\theta(\lambda-\kappa^2)G^*\big)^{-1}
=\kappa\;\!\Bigg\{-\sum_{j\in\N_\lambda}\tfrac{\v\;\!\P_j^\theta\v}{\vartheta_j(\kappa)}
+\kappa\Bigg(\u-\sum_{j\notin\N_\lambda}\tfrac{\v\;\!\P_j^\theta\v}
{\sqrt{(\lambda-\kappa^2-\lambda_j^\theta)^2-4}}\Bigg)\Bigg\}^{-1}.
$$
Moreover, as shown in Lemma \ref{lemma_sandwich}, the function
$$
O(\varepsilon)\ni\kappa\mapsto
\u-\sum_{j\notin\N_\lambda}\tfrac{\v\;\!\P_j^\theta\v}
{\sqrt{(\lambda-\kappa^2-\lambda_j^\theta)^2-4}}\in\B(\C^N)
$$
extends continuously to a function
$
\widetilde O(\varepsilon)\ni\kappa\mapsto M_1(\kappa)\in\B(\C^N)
$
with $\|M_1(\kappa)\|_{\B(\C^N)}$ uniformly bounded as $\kappa\to0$. Therefore, one
has for $\kappa\in O(\varepsilon)$
\begin{equation}\label{eq_I_0}
\big(\u+GR_0^\theta(\lambda-\kappa^2)G^*\big)^{-1}
=\kappa\;\!I_0(\kappa)^{-1}
\quad\hbox{with}\quad
I_0(\kappa):=-\sum_{j\in\N_\lambda}\tfrac{\v\;\!\P_j^\theta\v}{\vartheta_j(\kappa)}
+\kappa\;\!M_1(\kappa).
\end{equation}
Now, due to \eqref{eq_vartheta}, one has
\begin{equation}\label{eq_I_0_0}
I_0(0)
:=\lim_{\kappa\to0}I_0(\kappa)
=\tfrac12\sum_{\{j\mid\lambda=\lambda_j^\theta-2\}}\v\;\!\P_j^\theta\v
+\tfrac i2\sum_{\{j\mid\lambda=\lambda_j^\theta+2\}}\v\;\!\P_j^\theta\v,
\end{equation}
and since $I_0(0)$ has a positive imaginary part one infers from \cite[Cor.~2.8]{RT16}
that the orthogonal projection $S_0$ on $\ker\big(I_0(0)\big)$ is equal to the Riesz
projection of $I_0(0)$ associated with the value $0\in\sigma\big(I_0(0)\big)$, and
that the conditions (i)-(ii) of Proposition \ref{prop_inverse} hold. Applying this
proposition to $I_0(\kappa)$, one infers that for $\kappa\in\widetilde O(\varepsilon)$
with $\varepsilon>0$ small enough the operator $I_1(\kappa):S_0\C^N\to S_0\C^N$
defined by
\begin{equation}\label{eq_I_1}
I_1(\kappa)
:=\sum_{j\ge0}(-\kappa)^jS_0\left(M_1(\kappa)\big(I_0(0)+S_0\big)^{-1}\right)^{j+1}S_0
\end{equation}
is uniformly bounded as $\kappa\to0$. Furthermore, $I_1(\kappa)$ is invertible in
$S_0\C^N$ with bounded inverse satisfying
$$
I_0(\kappa)^{-1}
=\big(I_0(\kappa)+S_0\big)^{-1}+\tfrac1\kappa\big(I_0(\kappa)+S_0\big)^{-1}
S_0I_1(\kappa)^{-1}S_0\big(I_0(\kappa)+S_0\big)^{-1}.
$$
It follows that for $\kappa\in O(\varepsilon)$ with $\varepsilon>0$ small enough, one
has
\begin{equation}\label{eq_18}
\big(\u+GR_0^\theta(\lambda-\kappa^2)G^*\big)^{-1}
=\kappa\;\!\big(I_0(\kappa)+S_0\big)^{-1}
+\big(I_0(\kappa)+S_0\big)^{-1}S_0I_1(\kappa)^{-1}S_0\big(I_0(\kappa)+S_0\big)^{-1},
\end{equation}
with the first term vanishing as $\kappa\to0$. To describe the second term as
$\kappa\to0$ we note that the relation $\big(I_0(0)+S_0\big)^{-1}S_0=S_0$ and the
definition \eqref{eq_I_1} imply for $\kappa\in\widetilde O(\varepsilon)$ with
$\varepsilon>0$ small enough that
\begin{equation}\label{form_I_1}
I_1(\kappa)=S_0M_1(0)S_0+\kappa\;\!M_2(\kappa),
\end{equation}
with $M_1(0):=\lim_{\kappa\to0} M_1(\kappa)$ and
\begin{align*}
M_2(\kappa)
&:=-\tfrac1\kappa\;\!S_0\sum_{j\notin\N_\lambda}
\bigg(\tfrac1{\sqrt{(\lambda-\kappa^2-\lambda_j^\theta)^2-4}}
-\tfrac1{\sqrt{(\lambda-\lambda_j^\theta+i\;\!0)^2-4}}\bigg)
\v\;\!\P_j^\theta\v\;\!S_0\\
&\quad-\sum_{j\ge0}(-\kappa)^jS_0
\left(M_1(\kappa)\big(I_0(0)+S_0\big)^{-1}\right)^{j+2}S_0.
\end{align*}
Also, we note that the equality
\begin{equation}\label{eq_23}
\tfrac1{\sqrt{(\lambda-\kappa^2-\lambda_j^\theta)^2-4}}
=\tfrac1{\sqrt{\big((\lambda-\lambda_j^\theta)^2-4\big)
\left(1-\frac{2\kappa^2(\lambda-\lambda_j^\theta)-\kappa^4}
{(\lambda-\lambda_j^\theta)^2-4}\right)}},\quad j\notin\N_\lambda,
\end{equation}
implies that
$$
\lim_{\kappa\to0}\tfrac1{\sqrt{(\lambda-\kappa^2-\lambda_j^\theta)^2-4}}
=\tfrac1{\sqrt{(\lambda-\lambda_j^\theta+i\;\!0)^2-4}}
=\begin{cases}
-\beta_j^\theta(\lambda)^{-2} &\hbox{if $\lambda<\lambda_j^\theta-2$}\\
-i\beta_j^\theta(\lambda)^{-2} &\hbox{if $\lambda\in I^\theta_j$}\\
\beta_j^\theta(\lambda)^{-2} &\hbox{if $\lambda>\lambda_j^\theta+2$,}
\end{cases}
$$
and that $\|M_2(\kappa)\|_{\B(\C^N)}$ is uniformly bounded as $\kappa\to0$. Now, we
have
$$
M_1(0)
=\u+\sum_{\{j\mid\lambda<\lambda_j^\theta-2\}}\tfrac{\v\;\!\P_j^\theta\v}
{\beta_j^\theta(\lambda)^2}
+i\sum_{\{j\mid\lambda\in I^\theta_j\}}\tfrac{\v\;\!\P_j^\theta\v}
{\beta_j^\theta(\lambda)^2}
-\sum_{\{j\mid\lambda>\lambda_j^\theta+2\}}\tfrac{\v\;\!\P_j^\theta\v}
{\beta_j^\theta(\lambda)^2},
$$
with the sum over $\{j\mid\lambda>\lambda_j^\theta+2\}$ vanishing if $\lambda$
is a left threshold (i.e. $\lambda=\lambda_k^\theta-2$ for some $k$) and the sum over
$\{j\mid\lambda<\lambda_j^\theta-2\}$ vanishing if $\lambda$ is a right
threshold (i.e. $\lambda=\lambda_k^\theta+2$ for some $k$). Thus,
$I_1(0)=S_0M_1(0)S_0$ has a positive imaginary part. Therefore, the result
\cite[Cor.~2.8]{RT16} applies to the orthogonal projection $S_1$ on
$\ker\big(I_1(0)\big)$, and Proposition \ref{prop_inverse} can be applied to
$I_1(\kappa)$ as it was done for $I_0(\kappa)$. So, for
$\kappa\in\widetilde O(\varepsilon)$ with $\varepsilon>0$ small enough, the matrix
$I_2(\kappa):S_1\C^N\to S_1\C^N$ defined by
$$
I_2(\kappa)
:=\sum_{j\ge0}(-\kappa)^jS_1\left(M_2(\kappa)\big(I_1(0)+S_1\big)^{-1}\right)^{j+1}S_1
$$
is uniformly bounded as $\kappa\to0$. Furthermore, $I_2(\kappa)$ is invertible in
$S_1\C^N$ with bounded inverse satisfying
$$
I_1(\kappa)^{-1}
=\big(I_1(\kappa)+S_1\big)^{-1}+\tfrac1\kappa\;\!\big(I_1(\kappa)+S_1\big)^{-1}
S_1I_2(\kappa)^{-1}S_1\big(I_1(\kappa)+S_1\big)^{-1}.
$$
This expression for $I_1(\kappa)^{-1}$ can now be inserted in \eqref{eq_18} to get for
$\kappa\in O(\varepsilon)$ with $\varepsilon>0$ small enough
\begin{align}
&\big(\u+GR_0^\theta(\lambda-\kappa^2)G^*\big)^{-1}\nonumber\\
&=\kappa\;\!\big(I_0(\kappa)+S_0\big)^{-1}
+\big(I_0(\kappa)+S_0\big)^{-1}S_0\big(I_1(\kappa)+S_1\big)^{-1}S_0
\big(I_0(\kappa)+S_0\big)^{-1}\nonumber\\
&\quad+\tfrac1\kappa\;\!\big(I_0(\kappa)+S_0\big)^{-1}S_0
\big(I_1(\kappa)+S_1\big)^{-1}S_1I_2(\kappa)^{-1}S_1\big(I_1(\kappa)+S_1\big)^{-1}S_0
\big(I_0(\kappa)+S_0\big)^{-1},\label{eq_F_second}
\end{align}
with the first two terms bounded as $\kappa\to0$.

We now concentrate on the last term and check once more that the assumptions of
Proposition \ref{prop_inverse} are satisfied. For this, we recall that
$\big(I_1(0)+S_1\big)^{-1}S_1=S_1$, and observe that for
$\kappa\in\widetilde O(\varepsilon)$ with $\varepsilon>0$ small enough
\begin{equation}\label{eq_I_2}
I_2(\kappa)=S_1M_2(0)S_1+\kappa\;\!M_3(\kappa),
\end{equation}
with
$$
M_2(0)=-S_0M_1(0)\big(I_0(0)+S_0\big)^{-1}M_1(0)S_0
\quad\hbox{and}\quad
M_3(\kappa)\in\O(1).
$$
The inclusion $M_3(\kappa)\in\O(1)$ follows from simple computations taking the
expansion \eqref{eq_23} into account. As observed above, one has $M_1(0)=Y+iZ^*Z$,
with $Y,Z$ a hermitian matrices. Therefore,
$$
I_1(0)=S_0M_1(0)S_0=S_0YS_0+i\;\!(ZS_0)^*(ZS_0),
$$
and one infers from \cite[Cor.~2.5]{RT16} that $Z S_0S_1=0=S_1S_0 Z^*$. Since
$S_1S_0=S_1=S_0S_1$, it follows that $ZS_1=0=S_1Z^*$. Therefore, we have
$$
I_2(0)
=-S_1M_1(0)\big(I_0(0)+S_0\big)^{-1}M_1(0)S_1
=-S_1Y\big(I_0(0)+S_0\big)^{-1}YS_1,
$$
and since $I_0(0)+S_0=A+iB^*B$ with $A,B$ hermitian matrices (see \eqref{eq_I_0_0}) we
have
\begin{align}
\im\big(I_0(0)+S_0\big)^{-1}
&=\tfrac1{2i}\big\{(A+iB^*B)^{-1}-\big((A+iB^*B)^{-1}\big)^*\big\}\nonumber\\
&=\tfrac1{2i}(A+iB^*B)^{-1}(-2i)B^*B(A-iB^*B)^{-1}\nonumber\\
&=-\big(B(A-iB^*B)^{-1}\big)^*\big(B(A-iB^*B)^{-1}\big),\label{eq_AB}
\end{align}
from which we infer that
$
\im\big(I_2(0)\big)
=\im\big(-S_1Y\big(I_0(0)+S_0\big)^{-1}YS_1\big)\ge0
$.
So, the operator $I_2(0)$ satisfies the conditions of \cite[Cor.~2.8]{RT16}, and we
can once again apply Proposition \ref{prop_inverse} to $I_2(\kappa)$ with $S_2$ the
orthogonal projection on $\ker\big(I_2(0)\big)$. Thus, for
$\kappa\in\widetilde O(\varepsilon)$ with $\varepsilon>0$ small enough, the operator
$I_3(\kappa):S_2\C^N\to S_2\C^N$ defined by
$$
I_3(\kappa)
:=\sum_{j\ge0}(-\kappa)^jS_2\Big(M_3(\kappa)\big(I_2(0)+S_2\big)^{-1}\Big)^{j+1}S_2
$$
is uniformly bounded as $\kappa\to0$. Furthermore, $I_3(\kappa)$ is invertible in
$S_2\C^N$ with bounded inverse satisfying
$$
I_2(\kappa)^{-1}
=\big(I_2(\kappa)+S_2\big)^{-1}
+\tfrac1\kappa\big(I_2(\kappa)+S_2\big)^{-1}S_2I_3(\kappa)^{-1}S_2
\big(I_2(\kappa)+S_2\big)^{-1}.
$$
This expression for $I_2(\kappa)^{-1}$ can now be inserted in \eqref{eq_F_second} to
get for $\kappa\in O(\varepsilon)$ with $\varepsilon>0$ small enough
\begin{align}
&\big(\u+GR_0^\theta(\lambda-\kappa^2)G^*\big)^{-1}\nonumber\\
&=\kappa\big(I_0(\kappa)+S_0\big)^{-1}
+\big(I_0(\kappa)+S_0\big)^{-1}S_0\big(I_1(\kappa)+S_1\big)^{-1}S_0
\big(I_0(\kappa)+S_0\big)^{-1}\nonumber\\
&\quad+\tfrac1\kappa\big(I_0(\kappa)+S_0\big)^{-1}S_0
\big(I_1(\kappa)+S_1\big)^{-1}S_1\big(I_2(\kappa)+S_2\big)^{-1}S_1
\big(I_1(\kappa)+S_1\big)^{-1}S_0\big(I_0(\kappa)+S_0\big)^{-1}\nonumber\\
&\quad+\tfrac1{\kappa^2}\big(I_0(\kappa)+S_0\big)^{-1}S_0
\big(I_1(\kappa)+S_1\big)^{-1}S_1\big(I_2(\kappa)+ S_2\big)^{-1}S_2 I_3(\kappa)^{-1}
S_2\big(I_2(\kappa)+S_2\big)^{-1}S_1\nonumber\\
&\qquad\cdot\big(I_1(\kappa)+S_1\big)^{-1}S_0\big(I_0(\kappa)+S_0\big)^{-1}.
\label{eq_sol_1}
\end{align}
Fortunately, the iterative procedure stops here, as can be shown as in the proof of
\cite[Prop.~3.3]{RT17}. In consequence, the function
$$
O(\varepsilon)\ni\kappa\mapsto
\big(\u+GR^\theta_0(\lambda-\kappa^2)G^*\big)^{-1}\in\B(\C^N)
$$
extends continuously to a function
$
\widetilde O(\varepsilon)\ni\kappa\mapsto\M^\theta(\lambda,\kappa)\in\B(\C^N)
$,
with $\M^\theta(\lambda,\kappa)$ given by the r.h.s. of \eqref{eq_sol_1}.

(ii) Assume that $\lambda\in\sigma_{\rm p}(H^\theta)\setminus\T^\theta$, take
$\varepsilon>0$, let $\kappa\in\widetilde O(\varepsilon)$, and set
$J_0(\kappa):=T_0+\kappa^2T_1(\kappa)$ with
$$
T_0
:=\u+\sum_{\{j\mid\lambda<\lambda_j^\theta-2\}}\tfrac{\v\;\!\P_j^\theta\v}
{\beta_j^\theta(\lambda)^2}
+i\sum_{\{j\mid\lambda\in I^\theta_j\}}\tfrac{\v\;\!\P_j^\theta\v}
{\beta_j^\theta(\lambda)^2}
-\sum_{\{j\mid\lambda>\lambda_j^\theta+2\}}\tfrac{\v\;\!\P_j^\theta\v}
{\beta_j^\theta(\lambda)^2}
$$
and
$$
T_1(\kappa)
:=-\tfrac1{\kappa^2}\sum_{j\in\{1,\dots,N\}}
\bigg(\tfrac1{\sqrt{(\lambda-\kappa^2-\lambda_j^\theta)^2-4}}
-\tfrac1{\sqrt{(\lambda-\lambda_j^\theta+i\;\!0)^2-4}}\bigg)
\v\;\!\P^\theta_j\v.
$$
Then, one infers from \eqref{eq_23} that $\|T_1(\kappa)\|_{\B(\C^N)}$ is uniformly
bounded as $\kappa\to0$. Also, the assumptions of \cite[Cor.~2.8]{RT16} hold for the
operator $T_0$, and thus the orthogonal projection $S$ on $\ker(T_0)$ is equal to the
Riesz projection of $T_0$ associated with the value $0\in\sigma(T_0)$. It thus follows
from Proposition \ref{prop_inverse} that for $\kappa\in\widetilde O(\varepsilon)$ with
$\varepsilon>0$ small enough, the operator $J_1(\kappa):S\C^N\to S\C^N$ defined by
$$
J_1(\kappa)
:=\sum_{j\ge0}(-\kappa^2)^jS\;\!\big(T_1(\kappa)(T_0+S)^{-1}\big)^{j+1}S
$$
is uniformly bounded as $\kappa\to0$. Furthermore, $J_1(\kappa)$ is invertible in
$S\C^N$ with bounded inverse satisfying
$$
J_0(\kappa)^{-1}
=\big(J_0(\kappa)+S\big)^{-1}
+\tfrac1{\kappa^2}\big(J_0(\kappa)+S\big)^{-1}SJ_1(\kappa)^{-1}S
\big(J_0(\kappa)+S\big)^{-1}.
$$
It follows that for $\kappa\in O(\varepsilon)$ with $\varepsilon>0$ small enough one
has
\begin{equation}\label{eq_sol_2}
\big(\u+GR^\theta_0(\lambda-\kappa^2)G^*\big)^{-1}
=\big(J_0(\kappa)+S\big)^{-1}
+\tfrac1{\kappa^2}\big(J_0(\kappa)+S\big)^{-1}SJ_1(\kappa)^{-1}S
\big(J_0(\kappa)+S\big)^{-1}.
\end{equation}
The iterative procedure stops here, for the same reason as the one presented in the
proof of \cite[Prop.~3.3]{RT17} once we observe that
$$
J_1(\kappa)=ST_1(0)S+\kappa\;\!T_2(\kappa)
\quad\hbox{with}\quad T_2(\kappa)\in\O(1).
$$
Therefore, \eqref{eq_sol_2} implies that the function
$$
O(\varepsilon)\ni\kappa\mapsto
\big(\u+GR^\theta_0(\lambda-\kappa^2)G^*\big)^{-1}\in\B(\C^N)
$$
extends continuously to a function
$
\widetilde O(\varepsilon)\ni\kappa\mapsto
\M^\theta(\lambda,\kappa)\in\B(\C^N)
$,
with $\M^\theta(\lambda,\kappa)$ given by
\begin{equation}\label{eq_expansion_2}
\M^\theta(\lambda,\kappa)
=\big(J_0(\kappa)+S\big)^{-1}
+\tfrac1{\kappa^2}\big(J_0(\kappa)+S\big)^{-1}SJ_1(\kappa)^{-1}S
\big(J_0(\kappa)+S\big)^{-1}.\qedhere
\end{equation}
\end{proof}

\begin{Remark}\label{remark_no_accumu}
(a) A direct inspection shows that point (ii) of the proof of Proposition
\ref{Prop_asymp} applies in fact to all
$\lambda\in\sigma(H^\theta)\setminus\T^\theta$. So, the expansion \eqref{eq_sol_2}
holds for all $\lambda\in\sigma(H^\theta)\setminus\T^\theta$. Now, if
$\lambda\in\sigma(H^\theta)\setminus\big(\T^\theta\cup\sigma_{\rm p}(H^\theta)\big)$,
then the projection $S$ in point (ii) vanishes due to the definition of $T_0$ and
\eqref{eq_ker_T_0}. Therefore, for
$\lambda\in\sigma(H^\theta)\setminus\big(\T^\theta\cup\sigma_{\rm p}(H^\theta)\big)$,
the expansion \eqref{eq_sol_2} reduces to the equation
\begin{equation}\label{eq_case_3}
\big(\u+GR^\theta_0(\lambda-\kappa^2)G^*\big)^{-1}
=\big(T_0+\kappa^2T_1(\kappa)\big)^{-1}.
\end{equation}

(b) The asymptotic expansions of Proposition \ref{Prop_asymp} imply that the point
spectrum $\sigma_{\rm p}(H^\theta)$ is finite. Indeed, the eigenvalues of $H^\theta$
cannot accumulate at a point which is a threshold due to the expansion
\eqref{eq_sol_1} and the relation \eqref{eq_resolv_2}. They cannot accumulate at a
point which is an eigenvalue of $H^\theta$ due to the expansion \eqref{eq_sol_2} and
the relation \eqref{eq_resolv_2}. And finally, they cannot accumulate at a point of
$\sigma(H^\theta)\setminus\big(\T^\theta\cup\sigma_{\rm p}(H^\theta)\big)$ due to the
equation \eqref{eq_case_3} and the relation \eqref{eq_resolv_2}. Since the operator
$H^\theta$ is bounded, this implies that $\sigma_{\rm p}(H^\theta)$ is finite.
\end{Remark}

We close the section with some auxiliary results that can be deduced from the
expansions of Proposition \ref{Prop_asymp}. The notations are borrowed from the proof
of Proposition \ref{Prop_asymp} (with the only change that we extend by $0$ operators
defined originally on subspaces of $\C^N$ to get operators defined on all of $\C^N$),
and the proofs are given in the Appendix.

\begin{Lemma}\label{lemme_com}
Take $2\ge\ell\ge m\ge0$ and $\kappa\in\widetilde O(\varepsilon)$ with $\varepsilon>0$
small enough. Then, one has in $\B(\C^N)$
$$
\big[S_\ell,\big(I_m(\kappa)+S_m\big)^{-1}\big]\in\O(\kappa).
$$
\end{Lemma}

Given $\lambda\in\T^\theta$, we recall that
$\N_\lambda=\big\{j\in\{1,\dots,N\}\mid|\lambda-\lambda_j^\theta|=2\big\}$.

\begin{Lemma}\label{lemma_relations}
Let $\lambda\in\T^\theta$.
\begin{enumerate}
\item[(a)] For each $j\in\N_\lambda$, one has
$\P_j^\theta\v\;\!S_0=0=S_0\v\;\!\P_j^\theta$.
\item[(b)] For each $j\in\{1,\dots,n\}$ such that $\lambda\in I^\theta_j$, one has
$\P_j^\theta\v S_1=0=S_1\v\;\!\P_j^\theta$.
\item[(c)] One has $\re\big(M_1(0)\big)S_2=0=S_2\re\big(M_1(0)\big)$.
\item[(d)] One has $M_1(0)S_2=0=S_2M_1(0)$.
\end{enumerate}
\end{Lemma}

\subsection{Continuity of the scattering matrix}\label{sec_cont}

Based on the above asymptotic expansions, we now establish continuity properties of
the channel scattering matrices for the pair $(H^\theta,H^\theta_0)$ for each
$\theta\in[0,2\pi]$. Our approach is similar to that of \cite[Sec.~4]{RT17}, with one
major difference: Here, the scattering channels open at energies $\lambda_j^\theta-2$
and close at energies $\lambda_j^\theta+2$ for $j\in\{1,\dots,N\}$, while in
\cite{RT17} the scattering channels also open at specific energies but do no close
before reaching infinity.

Since  the stationary formula for the channel scattering matrix has already been
introduced in Section \ref{sec_intro}, we only need to provide more precise statements
about the continuity. Before presenting the result about the continuity at thresholds,
we define for each fixed $\lambda\in\T^\theta$, $\kappa\in\widetilde O(\varepsilon)$
with $\varepsilon>0$ small enough, and $2\ge\ell\ge m\ge0$, the operators
$$
C_{\ell m}(\kappa):=\big[S_\ell,\big(I_m(\kappa)+S_m\big)^{-1}\big]\in\B(\C^N),
$$
and note that $C_{\ell m}(\kappa)\in\O(\kappa)$ due to Lemma \ref{lemme_com}. In fact,
the formulas \eqref{eq_I_0}, \eqref{form_I_1} and \eqref{eq_I_2} imply that
$$
C_{\ell m}'(0):=\lim_{\kappa\to0}\tfrac1\kappa\;\!C_{\ell m}(\kappa)
$$
exists in $\B(\C^N)$. In other cases, we use the notation
$F(\kappa)\in\Oas(\kappa^n)$, $n\in\N$, for an operator $F(\kappa)\in\O(\kappa^n)$
such that $\lim_{\kappa\to0}\kappa^{-n}F(\kappa)$ exists in $\B(\C^N)$. We also note
that if $\kappa\in(0,\varepsilon)$ or $i\kappa\in(0,\varepsilon)$ with
$\varepsilon>0$, then $\kappa\in\widetilde O(\varepsilon)$ and
$-\kappa^2\in(-\varepsilon^2,\varepsilon^2)\setminus\{0\}$.

\begin{Theorem}\label{thm_cont}
Let $\lambda\in\T^\theta$, take $\kappa\in(0,\varepsilon)$ or
$i\kappa\in(0,\varepsilon)$ with $\varepsilon>0$ small enough, and let
$j,j'\in\{1,\dots,N\}$.
\begin{enumerate}
\item[(a)] If $\lambda\in I^\theta_j\cap I^\theta_{j'}$, then the limit
$\lim_{\kappa\to0}S^\theta(\lambda-\kappa^2)_{jj'}$ exists and is given by
$$
\lim_{\kappa\to0}S^\theta(\lambda-\kappa^2)_{jj'}
=\delta_{jj'}
-2i\;\!\beta_j^\theta(\lambda)^{-1}\P_j^\theta\v\;\!S_0\big(I_1(0)+S_1\big)^{-1}
S_0\v\;\!\P_{j'}^\theta\beta_{j'}^\theta(\lambda)^{-1}.
$$
\item[(b)] If $\lambda\in\overline{I^\theta_j}\cap\overline{I^\theta_{j'}}$ and
$-\kappa^2>0$, then the limit $\lim_{\kappa\to0}S^\theta(\lambda-\kappa^2)_{jj'}$
exists and is given by
$$
\lim_{\kappa\to0}S^\theta(\lambda-\kappa^2)_{jj'}
=\begin{cases}
0 &\hbox{if $\lambda>\lambda_j^\theta-2$, $\lambda=\lambda_{j'}^\theta-2$,}\\
0 &\hbox{if $\lambda=\lambda_j^\theta-2$, $\lambda>\lambda_{j'}^\theta-2$,}\\
\delta_{jj'}-\P_j^\theta\v\big(I_0(0)+S_0\big)^{-1}\v\;\!\P_{j'}^\theta\\
+\P_j^\theta\v\;\!C_{10}'(0)S_1\big(I_2(0)+S_2\big)^{-1}S_1C_{10}'(0)\v\;\!
\P_{j'}^\theta
&\hbox{if $\lambda_j^\theta-2=\lambda=\lambda_{j'}^\theta-2$.}
\end{cases}
$$
\item[(c)] If $\lambda\in\overline{I^\theta_j}\cap\overline{I^\theta_{j'}}$ and
$-\kappa^2<0$, then the limit
$\lim_{\kappa\to0}S^\theta(\lambda-\kappa^2)_{jj'}$ exists and is given by
$$
\lim_{\kappa\to0}S^\theta(\lambda-\kappa^2)_{jj'}
=\begin{cases}
0 &\hbox{if $\lambda<\lambda_j^\theta+2$, $\lambda=\lambda_{j'}^\theta+2$,}\\
0 &\hbox{if $\lambda=\lambda_j^\theta+2$, $\lambda<\lambda_{j'}^\theta+2$,}\\
\delta_{jj'}-i\;\!\P_j^\theta\v\big(I_0(0)+S_0\big)^{-1}\v\;\!\P_{j'}^\theta\\
+i\;\!\P_j^\theta\v\;\!C_{10}'(0)S_1\big(I_2(0)+S_2\big)^{-1}S_1C_{10}'(0)\v\;\!
\P_{j'}^\theta
&\hbox{if $\lambda_j^\theta+2=\lambda=\lambda_{j'}^\theta+2$.}
\end{cases}
$$
\end{enumerate}
\end{Theorem}

A detailed proof of this theorem is given in the Appendix. Let us however mention that
it is based on the following formula for the operator $\M^\theta(\lambda,\kappa)$,
which is obtained by rewriting the r.h.s. of \eqref{eq_sol_1} as in
\cite[Sec.~3.3]{RT16} and \cite[Sec.~4]{RT17}:
\begin{align}
&\M^\theta(\lambda,\kappa)\nonumber\\
&=\kappa\big(I_0(\kappa)+S_0\big)^{-1}\nonumber\\
&\quad+\Big(S_0\big(I_0(\kappa)+S_0\big)^{-1}-C_{00}(\kappa)\Big)S_0
\big(I_1(\kappa)+S_1\big)^{-1}S_0
\Big(\big(I_0(\kappa)+S_0\big)^{-1}S_0+C_{00}(\kappa)\Big)\nonumber\\
&\quad+\tfrac1\kappa\Big\{\Big(S_1\big(I_0(\kappa)+S_0\big)^{-1}-C_{10}(\kappa)\Big)
\big(I_1(\kappa)+ S_1\big)^{-1}-\Big(S_0\big(I_0(\kappa)+S_0\big)^{-1}
-C_{00}(\kappa)\Big)C_{11}(\kappa)\Big\}\nonumber\\
&\qquad\cdot S_1\big(I_2(\kappa)+S_2\big)^{-1}S_1\Big\{\big(I_1(\kappa)+S_1\big)^{-1}
\Big(\big(I_0(\kappa)+S_0\big)^{-1}S_1+C_{10}(\kappa)\Big)\nonumber\\
&\qquad+C_{11}(\kappa)\Big(\big(I_0(\kappa)+S_0\big)^{-1}S_0
+C_{00}(\kappa)\Big)\Big\}\nonumber\\
&\quad+\tfrac1{\kappa^2}\Big\{\Big[\Big(S_2\big(I_0(\kappa)+S_0\big)^{-1}
-C_{20}(\kappa)\Big)\big(I_1(\kappa)+S_1\big)^{-1}\nonumber\\
&\qquad-\Big(S_0\big(I_0(\kappa)+S_0\big)^{-1}-C_{00}(\kappa)\Big)C_{21}(\kappa)\Big]
\big(I_2(\kappa)+S_2\big)^{-1}\nonumber\\
&\qquad-\Big[\Big(S_1\big(I_0(\kappa)+S_0\big)^{-1}-C_{10}(\kappa)\Big)
\big(I_1(\kappa)+S_1\big)^{-1}\nonumber\\
&\qquad-\Big(S_0\big(I_0(\kappa)+S_0\big)^{-1}-C_{00}(\kappa)\Big)C_{11}(\kappa)\Big]
C_{22}(\kappa)\Big\}S_2I_3(\kappa)^{-1}S_2\nonumber\\
&\qquad\cdot\Big\{\big(I_2(\kappa)+S_2\big)^{-1}\Big[\big(I_1(\kappa)+S_1\big)^{-1}
\Big(\big(I_0(\kappa)+S_0\big)^{-1}S_2+C_{20}(\kappa)\Big)\nonumber\\
&\qquad+C_{21}(\kappa)\Big(\big(I_0(\kappa)+S_0\big)^{-1}S_0
+C_{00}(\kappa)\Big)\Big]\nonumber\\
&\qquad+C_{22}(\kappa)\Big[\big(I_1(\kappa)+S_1\big)^{-1}
\Big(\big(I_0(\kappa)+S_0\big)^{-1}S_1+C_{10}(\kappa)\Big)\nonumber\\
&\qquad+C_{11}(\kappa)\Big(\big(I_0(\kappa)+S_0\big)^{-1}S_0
+C_{00}(\kappa)\Big)\Big]\Big\}.\label{eq_grosse}
\end{align}
The interest of this formula is that the projections $S_\ell$ (which lead to
simplifications in the proof of the theorem) have been moved at the beginning or at
the end of each term.

Finally, we present the result about the continuity of the scattering matrix at
embedded eigenvalues not located at thresholds. As for the previous theorem, the proof
is given in the Appendix.

\begin{Theorem}\label{thm_cont_bis}
Let $\lambda\in\sigma_{\rm p}(H^\theta)\setminus\T^\theta$, take
$\kappa\in(0,\varepsilon)$ or $i\kappa\in(0,\varepsilon)$ with $\varepsilon>0$ small
enough, and let $j,j'\in\{1,\dots,N\}$. Then, if
$\lambda\in I^\theta_j\cap I^\theta_{j'}$, the limit
$\lim_{\kappa\to0}S^\theta(\lambda-\kappa^2)_{jj'}$ exists and is given by
\begin{equation}\label{eq_S_vp}
\lim_{\kappa\to0}S^\theta(\lambda-\kappa^2)_{jj'}
=\delta_{jj'}-2i\;\!\beta_j^\theta(\lambda)^{-1}\P_j^\theta\v\big(J_0(0)+S\big)^{-1}
\v\;\!\P_{j'}^\theta\beta_{j'}^\theta(\lambda)^{-1}.
\end{equation}
\end{Theorem}

\section{Structure of the wave operators}\label{sec_structure}
\setcounter{equation}{0}

In this section, we establish new stationary formulas both for the wave operators
$W^\theta_{\pm}$ at fixed value $\theta\in[0,2\pi]$ and the wave operators for the
initial pair of Hamiltonians $(H,H_0)$. First, we recall from \cite[Eq.~2.7.5]{Yaf92}
that $W_-^\theta$ satisfies for suitable $f,g\in\h$ the equation
$$
\big\langle W_-^\theta f,g\big\rangle_\h
=\int_\R\lim_{\varepsilon\searrow0}\tfrac\varepsilon\pi
\big\langle R_0^\theta(\lambda- i\varepsilon)f,R^\theta(\lambda- i\varepsilon)g
\big\rangle_\h\;\!\d\lambda.
$$
We also recall from \cite[Sec.~1.4]{Yaf92} that, given
$
\delta_\varepsilon(H_0^\theta-\lambda)
:=\frac{\pi^{-1}\varepsilon}{(H_0^\theta-\lambda)^2+\varepsilon^2}
$
with $\varepsilon>0$ and $\lambda\in\R$, the limit
$
\lim_{\varepsilon\searrow0}
\big\langle\delta_\varepsilon(H_0^\theta-\lambda)f,g\big\rangle_\h
$
exists for a.e. $\lambda\in\R$ and verifies the relation
$$
\langle f,g\rangle_\h
=\int_\R\lim_{\varepsilon\searrow0}\big\langle
\delta_\varepsilon(H_0^\theta-\lambda)f,g\big\rangle_\h\;\!\d\lambda.
$$
So, by taking \eqref{eq_resolv_1} into account and using the fact that
$\lim_{\varepsilon\searrow0}\|\delta_\varepsilon(H_0^\theta-\lambda)\|_{\B(\h)}=0$ if
$\lambda\notin\sigma(H^\theta_0)$, we obtain that
$$
\big\langle(W_-^\theta-1)f,g\big\rangle_\h
=-\int_{\sigma(H_0^\theta)}\lim_{\varepsilon\searrow0}\big\langle G^*
M^\theta(\lambda+i\varepsilon)G\delta_\varepsilon(H_0^\theta-\lambda)f,
R_0^\theta(\lambda-i\varepsilon)g\big\rangle_\h\;\!\d\lambda,
$$
with
$$
M^\theta(z):=\big(\u+GR_0^\theta(z)G^*\big)^{-1},\quad z\in\C\setminus\R.
$$

In the following sections, we derive an expression for the operator $(W_-^\theta-1)$
in the spectral representation of $H_0^\theta$; that is, for the operator
$\F^\theta(W_-^\theta-1)(\F^{\theta})^*$. For that purpose, we recall that
$G=\v\gamma_0$, with $\gamma_0:\h\to\C^N$ as in \eqref{def_gamma0}. We also define the
set
$$
\textstyle\Drond^\theta
:=\big\{\zeta\in\Hrond^\theta\mid\zeta=\sum_{j=1}^N\zeta_j,
~\zeta_j\in C^\infty_{\rm c}\big(I^\theta_j\setminus
\big(\T^\theta\cup\sigma_{\rm p}(H^\theta)\big);\P^\theta_j\C^N\big)\big\},
$$
which is dense in $\Hrond^\theta$ because $\T^\theta$ is countable and
$\sigma_{\rm p}(H^\theta)$ is closed and of Lebesgue measure $0$, as a consequence of
Remark \ref{remark_no_accumu}(b). Finally, we prove a small lemma useful for the
following computations:

\begin{Lemma}\label{lemma_short}
For $\zeta\in\Drond^\theta$ and $\lambda\in\sigma(H^\theta_0)$, one has
\begin{enumerate}
\item[(a)]
$
\gamma_0(\F^{\theta})^*\zeta
=\pi^{-1/2}\sum_{j\in\{1,\dots,N\}}\int_{I^\theta_j}\beta_j^\theta(\mu)^{-1}
\zeta_j(\mu)\;\!\d\mu\in\C^N
$,
\item[(b)]
$
\slim_{\varepsilon\searrow0}\gamma_0(\F^{\theta})^*
\delta_\varepsilon(X^\theta-\lambda)\zeta
=\pi^{-1/2}\sum_{\{j\mid\lambda\in I^\theta_j\}}\beta_j^\theta(\lambda)^{-1}
\zeta_j(\lambda)\in\C^N
$.
\end{enumerate}
\end{Lemma}

\begin{proof}
(a) follows from a simple computation taking \eqref{eq_adjoint} into account. For (b),
it is sufficient to note that the map $\mu\mapsto\beta_j^\theta(\mu)^{-1}\zeta_j(\mu)$
extends trivially to a continuous function on $\R$ with compact support in
$I_j^\theta$, and then to use the convergence of the Dirac delta sequence
$\delta_\varepsilon(\;\!\cdot\;\!-\lambda)$.
\end{proof}

Taking the previous observations into account, we obtain for
$\zeta,\xi\in\Drond^\theta$ the equalities
\begin{align}
&\big\langle\F^\theta(W_-^\theta-1)(\F^{\theta })^*\xi,
\zeta\big\rangle_{\Hrond^\theta}\nonumber\\
&=-\int_{\sigma(H^\theta_0)}\lim_{\varepsilon\searrow0}
\left\langle\gamma_0^*\v M^\theta(\lambda+i\varepsilon)\v\;\!\gamma_0(\F^{\theta})^*
\delta_\varepsilon(X^\theta-\lambda)\xi,
(\F^{\theta})^*\big(X^\theta-\lambda+i\varepsilon\big)^{-1}\zeta
\right\rangle_\h\;\!\d\lambda\nonumber\\
&=-\int_{\sigma(H_0^\theta)}\lim_{\varepsilon\searrow0}
\left\langle\v M^\theta(\lambda+i\varepsilon)\v\;\!\gamma_0(\F^{\theta})^*
\delta_\varepsilon(X^\theta-\lambda)\xi,
\gamma_0(\F^{\theta})^*\big(X^\theta-\lambda+i\varepsilon\big)^{-1}\zeta
\right\rangle_{\C^N}\;\!\d\lambda\nonumber\\
&=-\pi^{-1/2}\int_{\sigma(H^\theta_0)}\lim_{\varepsilon\searrow0}
\left\langle\v M^\theta(\lambda+i\varepsilon)\v\;\!\gamma_0(\F^{\theta})^*
\delta_\varepsilon(X^\theta-\lambda)\xi,
\sum_{j=1}^N\int_{I^\theta_j}\tfrac{\beta_j^\theta(\mu)^{-1}}{\mu-\lambda+i\varepsilon}
\;\!\zeta_j(\mu)\;\!\d\mu\right\rangle_{\C^N}\d\lambda\nonumber\\
&=-\pi^{-1/2}\sum_{j=1}^N\int_{I^\theta_j}\lim_{\varepsilon\searrow0}
\left\langle\v M^\theta(\lambda+i\varepsilon)\v\;\!\gamma_0(\F^{\theta})^*
\delta_\varepsilon(X^\theta-\lambda)\xi,
\int_{I_j^\theta}\tfrac{\beta_j^\theta(\mu)^{-1}}{\mu-\lambda+i\varepsilon}\;\!
\zeta_j(\mu)\;\!\d\mu\right\rangle_{\C^N}\d\lambda\label{eq_leading}\\
&\quad-\pi^{-1/2}\sum_{j=1}^N\int\limits_{\sigma(H^\theta_0)\setminus I^\theta_j}
\lim_{\varepsilon\searrow0}\left\langle\v M^\theta(\lambda+i\varepsilon)\v\;\!\gamma_0
(\F^{\theta})^*\delta_\varepsilon(X^\theta-\lambda)\xi,
\int_{I_j^\theta}\tfrac{\beta_j^\theta(\mu)^{-1}}{\mu-\lambda+i\varepsilon}\;\!
\zeta_j(\mu)\;\!\d\mu\right\rangle_{\C^N}\d\lambda.\label{eq_remainder}
\end{align}
In consequence, the expression for the operator
$\F^\theta(W_-^\theta-1)(\F^{\theta})^*$ is given by two terms, \eqref{eq_leading} and
\eqref{eq_remainder}, which we study separately in the next two sections.

\subsection{Main term of the wave operators}\label{sec_main}

We start this section with a key lemma which will allow us to rewrite the term
\eqref{eq_leading} in a rescaled energy representation, see \cite{BSB,IT19,SB16} for
similar constructions. For that purpose, we first define for $\theta\in[0,2\pi]$ and
$j\in\{1,\dots,N\}$ the unitary operator $\V_j^\theta:\ltwo(I_j^\theta)\to\ltwo(\R)$
given by
$$
\big(\V_j^\theta\xi\big)(s)
:=\tfrac{2^{1/2}}{\cosh(s)}\;\!\xi\big(\lambda_j^\theta+2\tanh(s)\big),
\quad\hbox{$\xi\in\ltwo(I_j^\theta)$, a.e. $s\in\R$,}
$$
with adjoint $(\V^{\theta}_j)^*:\ltwo(\R)\to\ltwo(I_j^\theta)$ given by
$$
\big((\V^{\theta}_j)^*f\big)(\lambda)
=\left(\tfrac2{4-(\lambda-\lambda_j^\theta)^2}\right)^{1/2}
f\Big(\arctanh\Big(\tfrac{\lambda-\lambda_j^\theta}2\Big)\Big),
\quad\hbox{$f\in\ltwo(\R)$, a.e. $\lambda\in I_j^\theta$.}
$$
Secondly, we define for any $\varepsilon>0$ the integral operator
$\Theta_{j,\varepsilon}^\theta$ on
$C_{\rm c}^\infty(I^\theta_j)\subset\ltwo(I_j^\theta)$ with kernel
$$
\Theta_{j,\varepsilon}^\theta(\lambda,\mu)
:=\tfrac i{2\pi(\mu-\lambda+i\varepsilon)}\;\!\beta^\theta_j(\lambda)\;\!
\beta_j^\theta(\mu)^{-1},\quad\lambda,\mu\in I_j^\theta.
$$
Thirdly, we write $D$ for the self-adjoint realisation of the operator
$-i\frac\d{\d s}$ in $\ltwo(\R)$ and $X$ for the operator of multiplication by the
variable in $\ltwo(\R)$. Finally, we define $b_\pm(X)\in\B\big(\ltwo(\R)\big)$ the
operators of multiplication by the bounded continuous functions
$$
b_\pm(s):=\big(\e^{s/2}\pm\e^{-s/2}\big)\big(\e^s+\e^{-s}\big)^{-1/2},\quad s\in\R,
$$
and note that $b_+$ is non-vanishing and satisfies $\lim_{s\to\pm\infty}b_+(s)=1$,
whereas $b_-$ vanishes at $s=0$ and satisfies $\lim_{s\to\pm\infty}b_-(s)=\pm 1$.

\begin{Lemma}\label{lemma_Pi}
For any $j\in\{1,\dots,N\}$, $f\in C^\infty_{\rm c}(\R)$ and $s\in\R$, one has
\begin{equation}\label{eq_kernel}
\lim_{\varepsilon\searrow0}\big(\V_j^\theta\Theta_{j,\varepsilon}^\theta
(\V_j^{\theta})^*f\big)(s)
=\big(\Pi(X,D)f\big)(s)
\end{equation}
with
$$
\Pi(X,D):=-\tfrac12\big(b_+(X)\tanh(\pi D)b_+(X)^{-1}
-ib_-(X)\cosh(\pi D)^{-1}b_+(X)^{-1}-1\big)\in\B\big(\ltwo(\R)\big).
$$
\end{Lemma}

The proof of this lemma consists in a direct computation together with the use of
formulas for the Fourier transforms of the functions appearing in the r.h.s. of
\eqref{eq_kernel}. Details are given in the Appendix.

\begin{Remark}\label{rem_form_Pi}
Since the functions appearing in $\Pi(X,D)$ have limits at $\pm\infty$, the operator
$\Pi(X,D)$ can be rewritten as
\begin{align*}
\Pi(X,D)
&=-\tfrac12\big(\tanh(\pi D)-i\tanh(X/2)\cosh(\pi D)^{-1}
+b_+(X)\big[\tanh(\pi D),b_+(X)^{-1}\big]\big)\\
&\quad-ib_-(X)\big[\cosh(\pi D)^{-1},b_+(X)^{-1}\big]-1\big)\\
&=-\tfrac12\big(\tanh(\pi D)-i\tanh(X)\cosh(\pi D)^{-1}-1\big)+K
\end{align*}
with
\begin{align*}
K&:=\tfrac i2\Big(\big(\tanh(X/2)-\tanh(X)\big)\cosh(\pi D)^{-1}
+ib_+(X)\big[\tanh(\pi D),b_+(X)^{-1}\big]\\
&\quad+b_-(X)\big[\cosh(\pi D)^{-1},b_+(X)^{-1}\big]\Big)\in\K\big(\ltwo(\R)\big).
\end{align*}
See for example \cite[Thm.~4.1]{Sim05} and \cite[Thm.~C]{Cor75} for a justification of
the compactness of the operator $K$. Note also that an operator similar to $\Pi(X,D)$
already appeared in \cite{IT19} in the context of potential scattering on the discrete
half-line.
\end{Remark}

Now, define the unitary operator $\V^\theta:\Hrond^\theta\to\ltwo(\R;\C^N)$ by
$$
\V^\theta\zeta
:=\sum_{j=1}^N\big(\V^\theta_j\otimes\P_j^\theta\big)\zeta|_{I^\theta_j},
\quad\zeta\in\Hrond^\theta,
$$
with adjoint $(\V^{\theta})^*:\ltwo(\R;\C^N)\to\Hrond^\theta$ given by
$$
\big((\V^{\theta})^*f\big)(\lambda)
=\sum_{\{j\mid\lambda\in I^\theta_j\}}
\left(\big((\V^{\theta}_j)^*\otimes\P_j^\theta\big)f\right)(\lambda),
\quad f\in\ltwo(\R;\C^N),~\hbox{a.e. $\lambda\in I^\theta$,}
$$
and for any $\varepsilon>0$, define the integral operator $\Theta^\theta_\varepsilon$
on $\Drond^\theta\subset\Hrond^\theta$ by
$$
\Theta^\theta_\varepsilon\zeta
:=\sum_{j=1}^N\big(\Theta^\theta_{j,\varepsilon}\otimes1_N\big)\zeta_j,
\quad\zeta=\sum_{j=1}^N\zeta_j\in\Drond^\theta.
$$
Then, using the results that precede, one obtains a simpler expression for the term
\eqref{eq_leading}:

\begin{Proposition}\label{prop_main_term}
For any $\xi\in\Drond^\theta$ and $f\in\V^\theta\Drond^\theta$, one has the equality
\begin{align*}
&-\pi^{-1/2}\sum_{j=1}^N\int_{I^\theta_j}\lim_{\varepsilon\searrow0}
\bigg\langle\v M^\theta(\lambda+i\varepsilon)\v\;\!\gamma_0(\F^{\theta})^*
\delta_\varepsilon(X^\theta-\lambda)\xi,
\int_{I_j^\theta}\tfrac{\beta_j^\theta(\mu)^{-1}}{\mu-\lambda+i\varepsilon}\;\!
\big((\V^{\theta})^*f\big)_j(\mu)\;\!\d\mu\bigg\rangle_{\C^N}\d\lambda\\
&=\left\langle(\V^{\theta})^*\big(\Pi(X,D)^*\otimes1_N\big)\V^\theta
\big(S^\theta(X^\theta)-1\big)\xi,(\V^{\theta})^*f\right\rangle_{\Hrond^\theta}
\end{align*}
with $S^\theta(X^\theta):=\F^\theta S^\theta(\F^{\theta})^*$.
\end{Proposition}

\begin{proof}
Using Lemma \ref{lemma_short}(b), Lemma \ref{lemma_Pi}, and \eqref{eq_S_lambda}, one
obtains
\begin{align*}
&-\pi^{-1/2}\sum_{j=1}^N\int_{I^\theta_j}\lim_{\varepsilon\searrow0}
\bigg\langle\v M^\theta(\lambda+i\varepsilon)\v\;\!\gamma_0(\F^{\theta})^*
\delta_\varepsilon(X^\theta-\lambda)\xi,
\int_{I_j^\theta}\tfrac{\beta_j^\theta(\mu)^{-1}}{\mu-\lambda+i\varepsilon}
\big((\V^{\theta})^*f\big)_j(\mu)\;\!\d\mu\bigg\rangle_{\C^N}\d\lambda\\
&=-2i\pi^{1/2}\sum_{j=1}^N\int_{I^\theta_j}\lim_{\varepsilon\searrow0}
\bigg\langle\beta_j^{\theta}(\lambda)^{-1}\v M^\theta(\lambda+i\varepsilon)
\v\;\!\gamma_0(\F^{\theta})^*\delta_\varepsilon(X^\theta-\lambda)\xi,\\
&\quad\left(\tfrac2{4-(\lambda-\lambda_j^\theta)^2}\right)^{1/2}
\big(\big(\V^\theta_j\Theta_{j,\varepsilon}^\theta(\V^{\theta}_j)^*
\otimes\P_j^\theta\big)f\big)
\Big(\arctanh\Big(\tfrac{\lambda-\lambda_j^\theta}2\Big)\Big)
\bigg\rangle_{\C^N}\;\!\d\lambda\\
&=-2i\sum_{j=1}^N\int_{I^\theta_j}\Bigg\langle\beta_j^{\theta}(\lambda)^{-1}
\v\;\!\M^\theta(\lambda,0)\v\;\!\sum_{\{k\mid\lambda\in I^\theta_k\}}
\beta_k^\theta(\lambda)^{-1}\xi_k(\lambda),\\
&\quad\left(\tfrac2{4-(\lambda-\lambda_j^\theta)^2}\right)^{1/2}
\big(\big(\Pi(X,D)\otimes\P_j^\theta\big)f\big)
\Big(\arctanh\Big(\tfrac{\lambda-\lambda_j^\theta}2\Big)\Big)
\Bigg\rangle_{\C^N}\;\!\d\lambda\\
&=\sum_{j=1}^N\int_{I^\theta_j}\Bigg\langle-2i\sum_{\{k\mid\lambda\in I^\theta_k\}}
\beta_j^{\theta}(\lambda)^{-1}\P_j^\theta\v\;\!\M^\theta(\lambda,0)
\v\;\!\P_k^\theta\beta_k^\theta(\lambda)^{-1}\xi_k(\lambda),\\
&\quad\left(\tfrac2{4-(\lambda-\lambda_j^\theta)^2}\right)^{1/2}
\big(\big(\Pi(X,D)\otimes1_N\big)f\big)
\Big(\arctanh\Big(\tfrac{\lambda-\lambda_j^\theta}2\Big)\Big)
\Bigg\rangle_{\C^N}\;\!\d\lambda\\
&=\sum_{j=1}^N\int_{I^\theta_j}\left\langle\P_j^\theta\big(S^\theta(\lambda)-1\big)
\xi(\lambda),\left(\tfrac2{4-(\lambda-\lambda_j^\theta)^2}\right)^{1/2}
\big(\big(\Pi(X,D)\otimes1_N\big)f\big)
\Big(\arctanh\Big(\tfrac{\lambda-\lambda_j^\theta}2\Big)\Big)
\right\rangle_{\C^N}\d\lambda\\
&=\sum_{j=1}^N\int_\R\left\langle\tfrac{2^{1/2}}{\cosh(s)}\;\!\P_j^\theta
\big(\big(S^\theta(X^\theta)-1\big)\xi\big)\big(\lambda_j^\theta+2\tanh(s)\big),
\big(\big(\Pi(X,D)\otimes1_N\big)f\big)(s)\right\rangle_{\C^N}\d s\\
&=\left\langle\V^\theta\big(S^\theta(X^\theta)-1\big)\xi,
\big(\Pi(X,D)\otimes1_N\big)f\right\rangle_{\ltwo(\R;\C^N)}\\
&=\left\langle(\V^{\theta})^*\big(\Pi(X,D)^*\otimes1_N\big)\V^\theta
\big(S^\theta(X^\theta)-1\big)\xi,(\V^{\theta})^*f\right\rangle_{\Hrond^\theta}
\end{align*}
as desired.
\end{proof}

Remark \ref{rem_form_Pi} and Proposition \ref{prop_main_term} imply that the operator
defined by \eqref{eq_leading} extends continuously to the operator
\begin{equation}\label{eq_form_leading}
-\tfrac12(\V^{\theta})^*\big\{\big(\tanh(\pi D)
+i\cosh(\pi D)^{-1}\tanh(X)-1\big)\otimes1_N\big\}\V^\theta
\big(S^\theta(X^\theta)-1\big)+K^\theta\in\B(\Hrond^\theta)
\end{equation}
with
\begin{equation}\label{eq_compact_leading}
K^\theta:=(\V^{\theta})^*(K^*\otimes1_N)\;\!\V^\theta\big(S^\theta(X^\theta)-1\big)
\in\K(\Hrond^\theta).
\end{equation}

\subsection{Remainder term of the wave operators}

We prove in this section that the operator defined by the remainder term
\eqref{eq_remainder} in the expression for $(W_-^\theta-1)$ extends to a compact
operator under generic conditions. In the next section, we deal with the remaining
exceptional cases. Our proof is based on two lemmas. The first lemma complements the
continuity properties established in Section \ref{sec_cont}, and it is similar to
Lemma 5.3 of \cite{RT17} in the continuous setting. Its technical proof is given in
the Appendix.

\begin{Lemma}\label{lemma_continuity}
For any $\theta\in[0,2\pi]$ and $j,j'\in\{1,\dots,N\}$ such that
$\lambda_j^\theta\ne\lambda_{j'}^\theta$, the function
\begin{equation}\label{eq_def_function}
I^\theta_{j'}\setminus\big(I_j^\theta\cup\T^\theta\cup\sigma_{\rm p}(H^\theta)\big)
\ni\lambda\mapsto\beta_j^\theta(\lambda)^{-2}\;\!\P_j^\theta\v\;\!\M^\theta(\lambda,0)
\v\;\!\P_{j'}^\theta\in\B(\C^N)
\end{equation}
extends to a continuous function on $\overline{I^\theta_{j'}\setminus I_j^\theta}$.
\end{Lemma}

The second lemma deals with a factor in the remainder term \eqref{eq_remainder}. For
its proof (also given in the Appendix) and for later use, we recall that the dilation
group $\{V_t\}_{t\in \R}$ given by $V_tf:=\e^{t/2}f(\e^t\;\!\cdot\;\!)$,
$f\in\ltwo\big((0,\infty)\big)$, is a strongly continuous unitary group in
$\ltwo\big((0,\infty)\big)$ with self-adjoint generator denoted by $A_+$.

\begin{Lemma}\label{lemma_compact}
Let $\theta\in[0,2\pi]$ and $j,j'\in\{1,\dots,N\}$ be such that either
$\lambda_{j'}^\theta<\lambda_j^\theta$ and $\lambda_j^\theta-2<\lambda_{j'}^\theta+2$,
or $\lambda_j^\theta<\lambda_{j'}^\theta$ and
$\lambda_{j'}^\theta-2<\lambda_j^\theta+2$. Then, the integral operator $\vartheta$ on
$C_{\rm c}^\infty(I_j^\theta)\subset\ltwo(I_j^\theta)$ given by
$$
(\vartheta f)(\lambda)
:=\int_{I_j^\theta}\tfrac1{\mu-\lambda}\;\!\beta_j^\theta(\lambda)^2
\beta_{j'}^\theta(\lambda)^{-1}\beta_j^\theta(\mu)^{-1}f(\mu)\;\!\d\mu,
\quad f\in C_{\rm c}^\infty(I_j^\theta),~\lambda\in I_{j'}^\theta\setminus I_j^\theta,
$$
extends continuously to a compact operator from $\ltwo(I_j^\theta)$ to
$\ltwo(I_{j'}^\theta\setminus I_j^\theta)$.
\end{Lemma}

\begin{Remark}\label{remark_bad}
The proof of Lemma \ref{lemma_compact} does not work in the exceptional cases
$\lambda_{j'}^\theta<\lambda_j^\theta$, $\lambda_j^\theta-2=\lambda_{j'}^\theta+2$ and
$\lambda_j^\theta<\lambda_{j'}^\theta$, $\lambda_{j'}^\theta-2=\lambda_j^\theta+2$.
These exceptional cases will be discussed in the next
section. A direct inspection shows that the condition
$\lambda_j^\theta-2=\lambda_{j'}^\theta+2$ is verified if and only if $\theta=0$,
$N\in2\N$, and $(j,j')=(N,N/2)$ (it can also be verified for some $N$, $j$, and $j'$
when $\theta=2\pi$, but this gives nothing new since the cases $\theta=0$ and
$\theta=2\pi$ are equivalent, see Remark \ref{rem_A_theta}(a)).
\end{Remark}

Putting together the results of both lemmas leads to the compactness of the operator
defined by the remainder term \eqref{eq_remainder}:

\begin{Proposition}\label{prop_NUS2}
If $\theta\ne0$ or $N\notin2\N$, then the operator defined by \eqref{eq_remainder}
extends continuously to a compact operator in $\Hrond^\theta$.
\end{Proposition}

\begin{proof}
Let $\xi,\zeta\in\Drond^\theta$. Then, Lemma \ref{lemma_short}(b) implies that
\eqref{eq_remainder} is equal to
\begin{align*}
&-\pi^{-1/2}\sum_{j=1}^N\int\limits_{\sigma(H^\theta_0)\setminus I^\theta_j}
\lim_{\varepsilon\searrow0}\bigg\langle\v M^\theta(\lambda+i\varepsilon)\v\;\!\gamma_0
(\F^{\theta})^*\delta_\varepsilon(X^\theta-\lambda)\xi,
\int_{I_j^\theta}\tfrac{\beta_j^\theta(\mu)^{-1}}{\mu-\lambda+i\varepsilon}\;\!
\zeta_j(\mu)\;\!\d\mu\bigg\rangle_{\C^N}\d\lambda\\
&=\pi^{-1}\sum_{j=1}^N\int\limits_{\sigma(H^\theta_0)\setminus I^\theta_j}\bigg\langle\v
\M^\theta(\lambda,0)\v\;\!\sum_{\{k\mid\lambda\in I^\theta_k\}}
\beta_k^\theta(\lambda)^{-1}\xi_k(\lambda),\int_{I_j^\theta}
\tfrac{\beta_j^\theta(\mu)^{-1}}{\mu-\lambda}\;\!\zeta_j(\mu)\;\!\d\mu
\bigg\rangle_{\C^N}\d\lambda\\
&=\pi^{-1}\sum_{j,\,j'=1}^N\int\limits_{I_{j'}^\theta\setminus I^\theta_j}\bigg\langle\v
\M^\theta(\lambda,0)\v\;\!\beta_{j'}^\theta(\lambda)^{-1}\xi_{j'}(\lambda),
\int_{I_j^\theta}\tfrac{\beta_j^\theta(\mu)^{-1}}{\mu-\lambda}\;\!\zeta_j(\mu)\;\!\d\mu
\bigg\rangle_{\C^N}\d\lambda\\
&=\pi^{-1}\sum_{j,\,j'=1}^N\int\limits_{I_{j'}^\theta\setminus I^\theta_j}\bigg\langle
\beta_j^\theta(\lambda)^{-2}\P_j^\theta\v\M^\theta(\lambda,0)\v\;\!\P_{j'}^\theta
\xi_{j'}(\lambda),\int_{I_j^\theta}\tfrac{\beta_j^\theta(\lambda)^2
\beta_{j'}^\theta(\lambda)^{-1}\beta_j^\theta(\mu)^{-1}}{\mu-\lambda}\;\!\zeta_j(\mu)
\;\!\d\mu\bigg\rangle_{\C^N}\d\lambda.
\end{align*}
Now, we know from Lemma \ref{lemma_continuity} that the function
\begin{equation}\label{eq_def_n}
I^\theta_{j'}\setminus\big(I_j^\theta\cup\T^\theta\cup\sigma_{\rm p}(H^\theta)\big)
\ni\lambda \mapsto n_{j,j'}^\theta(\lambda)
:=\beta_j^\theta(\lambda)^{-2}\P_j^\theta\v\M^\theta(\lambda,0)\v\;\!\P_{j'}^\theta
\end{equation}
extends to a continuous function on $\overline{I_{j'}^\theta\setminus I_j^\theta}$. We
denote by $N_{j,j'}^\theta$ the corresponding bounded multiplication operator from
$\ltwo\big(I_{j'}^\theta\setminus I_j^\theta;\P_{j'}^\theta\C^N\big)$ to
$\ltwo\big(I_{j'}^\theta\setminus I_j^\theta;\P_{j}^\theta\C^N\big)$. Furthermore,
Lemma \ref{lemma_compact} implies that the integral operator $\vartheta_{j,j'}^\theta$
on
$
C^\infty_{\rm c}\big(I^\theta_j\setminus
\big(\T^\theta\cup\sigma_{\rm p}(H^\theta)\big);\P^\theta_j\C^N\big)
$
given by
$$
\big(\vartheta_{j,j'}^\theta\zeta_j\big)(\lambda)
:=\int_{I_j^\theta}\tfrac{\beta_j^\theta(\lambda)^2\beta_{j'}^\theta(\lambda)^{-1}
\beta_j^\theta(\mu)^{-1}}{\mu-\lambda}\;\!\zeta_j(\mu)\;\!\d\mu,
\quad\zeta_j\in C^\infty_{\rm c}\big(I^\theta_j\setminus
\big(\T^\theta\cup\sigma_{\rm p}(H^\theta)\big);\P^\theta_j\C^N\big),
~\lambda\in I_{j'}^\theta\setminus I_j^\theta,
$$
extends continuously to a compact operator from
$\ltwo\big(I_j^\theta;\P_j^\theta\C^N\big)$ to
$\ltwo\big(I_{j'}^\theta\setminus I_j^\theta;\P_j^\theta\C^N\big)$. Therefore, we
obtain that
\begin{align*}
\eqref{eq_remainder}
&=\pi^{-1}\sum_{j,\,j'=1}^N\int_{I_{j'}^\theta\setminus I^\theta_j}
\left\langle n_{j,j'}^\theta(\lambda)\xi_{j'}(\lambda),
\big(\vartheta_{j,j'}^\theta\zeta_j\big)(\lambda)\right\rangle_{\C^N}\d\lambda\\
&=\pi^{-1}\sum_{j,\,j'=1}^N\left\langle N_{j,j'}^\theta(1_{j,j'}^\theta)^*
\xi_{j'}, \vartheta_{j,j'}^\theta\zeta_j
\right\rangle_{\ltwo(I_{j'}^\theta\setminus I_j^\theta;\P_j^\theta\C^N)}\\
&=\sum_{j=1}^N\left\langle\pi^{-1}\sum_{j'=1}^N(\vartheta_{j,j'}^\theta)^*
N_{j,j'}^\theta(1_{j,j'}^\theta)^*\xi_{j'},\zeta_j
\right\rangle_{\ltwo(I_j^\theta;\P_j^\theta\C^N)}\\
&=\left\langle k^\theta\xi,\zeta\right\rangle_{\oplus_{j=1}^N
\ltwo(I_j^\theta;\P_j^\theta\C^N)},
\end{align*}
with $1_{j,j'}^\theta$ the inclusion of
$\ltwo\big(I_{j'}^\theta\setminus I_j^\theta;\P_{j'}^\theta\C^N\big)$ into
$\ltwo\big(I_{j'}^\theta;\P_{j'}^\theta\C^N\big)$ and $k^\theta$ the compact operator
from $\Hrond^\theta$ to $\oplus_{j=1}^N\ltwo(I_j^\theta;\P_j^\theta\C^N)$ given by
$$
\big(k^\theta\xi\big)_j
:=\pi^{-1}\sum_{j'=1}^N(\vartheta_{j,j'}^\theta)^*N_{j,j'}^\theta
(1_{j,j'}^\theta)^*\xi_{j'},\quad\xi\in\Hrond^\theta,~j\in\{1,\dots,N\}.
$$
Since the Hilbert spaces $\oplus_{j=1}^N\ltwo(I_j^\theta;\P_j^\theta\C^N)$ and
$\Hrond^\theta$ are isomorphic, this implies that the operator defined by
\eqref{eq_remainder} extends continuously to a compact operator in $\Hrond^\theta$.
\end{proof}

\subsection{Exceptional case}\label{sec_except}

In this section, we consider the exceptional cases
$\lambda_j^\theta-2=\lambda_{j'}^\theta+2$ and
$\lambda_{j'}^\theta-2=\lambda_{j}^\theta+2$, which take place for the values
$\theta=0$, $N\in2\N$, $(j,j')=(N,N/2)$ (first case) and $(j,j')=(N/2,N)$ (second
case). As mentioned in Remark \ref{remark_bad}, the proof of Lemma \ref{lemma_compact}
does not work in these cases, and therefore one cannot infer that the operator defined
by remainder term \eqref{eq_remainder} is compact. Further analysis is necessary, and
this is precisely the content of this section.

First, we recall from Remark \ref{rem_A_theta} that the eigenvalues
$\lambda_{N/2}^0,\lambda_N^0\in\R$ of $A^0$ are $\lambda_{N/2}^0=-2$, $\lambda_N^0=2$
and the eigenvectors $\xi_{N/2}^0,\xi_N^0\in\C^N$ of $A^0$ have components
\begin{equation}\label{eq_vectors}
\big(\xi_{N/2}^0\big)_j=(-1)^j
\quad\hbox{and}\quad
\big(\xi_N^0)_j=1,\quad j\in\{1,\dots,N\}.
\end{equation}
Also, we define $\L:=\sp\{\v\;\!\xi_{N/2}^0,\v\;\!\xi_N^0\}$, and note that $\L$ is a
subspace of $\C$ of (complex) dimension $1$ or $2$ because $\v\ne0$. We start by
determining the range of the projections $S_0,S_1,S_2$ appearing in the asymptotic
expansion \eqref{eq_sol_1} when $\lambda=0:$

\begin{Lemma}\label{lemma_exceptional}
Let $\theta=0$, $N\in2\N$, and $\lambda=0$. Then, $S_0\C^N=\L^\bot$ and $S_1=S_2=0$.
\end{Lemma}

\begin{proof}
(i) Since
\begin{equation}\label{eq_I(0)}
2\;\!I_0(0)
=\sum_{\{j\mid\lambda_j^\theta=2\}}\v\;\!\P_j^0\v
+i\sum_{\{j\mid\lambda_j^\theta=-2\}}\v\;\!\P_j^0\v
=\v\;\!\P^0_N\v+i\;\!\v\;\!\P^0_{N/2}\v,
\end{equation}
one infers that $\xi\in\ker\big(I_0(0)\big)$ if and only if $\xi\bot\v\xi^0_N$ and
$\xi\bot\v\xi^0_{N/2}$, which shows that $S_0\C^N=\L^\bot$.

(ii) One has $\xi\in S_1\C^N$ if and only if $\xi\in S_0\C^N$ and
$\xi\in\ker\big(S_0 M_1(0)S_0\big)$. But, since $\lambda_j^0-2=0\Leftrightarrow j=N$
and $\lambda_j^0+2=0\Leftrightarrow j=N/2$, we have
$$
S_0M_1(0)S_0
=S_0\u S_0+i\sum_{\{j\mid0\in I^0_j\}}
\tfrac1{\beta_j^0(0)^2}\;\!S_0\v\;\!\P^0_j\v S_0
=S_0\u S_0+i\sum_{j\ne N/2,N}\tfrac1{\beta_j^0(0)^2}\;\!S_0\v\;\!\P^0_j\v S_0.
$$
Therefore, $\xi\in S_1\C^N$ if and only if $\xi\in S_0\C^N$, $\xi\in\ker(S_0\u S_0)$,
and $\xi\in\ker(S_0\v\;\!\P^0_j\v S_0)$ for all $j\ne N/2,N$. Due to point (i), these
conditions hold if and only if $\xi\in\L^\bot$, $\u\xi\in\L$, and
$\v\xi=c_1\xi^0_{N/2}+c_2\xi^0_N$ for some $c_1,c_2\in\C$. But, since the first and
third conditions are equivalent to $\v\xi=0$, these three conditions reduce to
$\xi\in\u\L$ and $\v\xi=0$. Finally, a direct inspection taking into account the
formulas \eqref{eq_vectors} shows that this can be satisfied only if $\xi=0$. Thus
$S_1=0$, and so $S_2=0$ too.
\end{proof}

Now, consider the function studied in Lemma \ref{lemma_continuity} for the values
$\theta=0$, $N\in2\N$, $(j,j')=(N,N/2)$, when $\lambda\nearrow0$. Using the notation
$\lambda=0-\kappa^2$ with $\kappa>0$ and the equalities
$\P^0_N\v S_0=0=S_0\v\P^0_{N/2}$ and $S_1=S_2=0$ from Lemma \ref{lemma_relations}(a)
and Lemma \ref{lemma_exceptional}, we infer from \eqref{eq_grosse} that
\begin{align*}
&\beta^0_N(-\kappa^2)^{-2}\;\!\P^0_N\v\;\!\M^0(0,\kappa)\v\;\!\P^0_{N/2}\\
&=\kappa\;\!\beta^0_N(-\kappa^2)^{-2}\P^0_N\v\left(\big(I_0(\kappa)+S_0\big)^{-1}
-\tfrac1\kappa\;\!C_{00}(\kappa)S_0\;\!I_1(\kappa)^{-1}S_0C_{00}(\kappa)\right)
\v\;\!\P^0_{N/2}.
\end{align*}
Using then the expansions
$$
\kappa\;\!\beta^0_N(-\kappa^2)^{-2}=1/2+\O(\kappa^2),\quad
\big(I_0(\kappa)+S_0\big)^{-1}\v\;\!\P^0_{N/2}=I_0(0)^{-1}\v\;\!\P^0_{N/2}+\O(\kappa),
\quad C_{00}(\kappa)\in\O(\kappa),
$$
we get
\begin{equation}\label{eq_Wuhan3}
\beta^0_N(-\kappa^2)^{-2}\;\!\P^0_N\v\;\!\M^0(0,\kappa)\v\;\!\P^0_{N/2}
=\tfrac12\;\!\P_N^0\v I_0(0)^{-1}\v\;\!\P_{N/2}^0+\O(\kappa).
\end{equation}

Similarly, for the values $\theta=0$, $N\in2\N$, $(j,j')=(N/2,N)$, when
$\lambda\searrow0$, using the notation $\lambda=0-\kappa^2$ with $i\kappa>0$ and
arguments as above, we infer from \eqref{eq_grosse} that
\begin{align*}
&\beta^0_{N/2}(-\kappa^2)^{-2}\;\!\P^0_{N/2}\v\;\!\M^0(0,\kappa)\v\;\!\P^0_N\\
&=\kappa\;\!\beta^0_{N/2}(-\kappa^2)^{-2}\P^0_{N/2}\v
\left(\big(I_0(\kappa)+S_0\big)^{-1}-\tfrac1\kappa\;\!C_{00}(\kappa)S_0
\;\!I_1(\kappa)^{-1}S_0C_{00}(\kappa)\right)\v\;\!\P^0_N\\
&=-\tfrac i2\;\!\P_{N/2}^0\v I_0(0)^{-1}\v\;\!\P_N^0+\O(\kappa).
\end{align*}

\begin{Lemma}\label{lemma_independent}
Let $\theta=0$, $N\in2\N$, and $\lambda=0$. Then, the following statements are
equivalent:
\begin{enumerate}
\item[(i)] The vectors $\v\xi^0_N$ and $\v\xi^0_{N/2}$ are linearly independent,
\item[(ii)] $\P_N^0 \v I_0(0)^{-1}\v\;\!\P_{N/2}^0=0$,
\item[(iii)] $\P_{N/2}^0 \v I_0(0)^{-1}\v\;\!\P_N^0=0$.
\end{enumerate}
\end{Lemma}

\begin{proof}
(i)$\Rightarrow$(ii) Since $\v\xi^0_N$ and $\v\xi^0_{N/2}$ are linearly independent,
there exists an invertible operator $B\in\B(\L,\C^2)$ such that
$
B\v\xi^0_N=\left(\begin{smallmatrix}1\\ 0\end{smallmatrix}\right)$ and
$B\v\xi^0_{N/2}=\left(\begin{smallmatrix}0\\1\end{smallmatrix}\right)$.
It thus follows from \eqref{eq_I(0)} that
$$
2 B \;\!I_0(0) B^*
=\big(B\v\;\!\P^0_N\big)\big(B\v\;\!\P^0_N\big)^*
+i\big(B\v\;\!\P^0_{N/2}\big)\big(B\v\;\!\P^0_{N/2}\big)^*
=\begin{pmatrix}1&0\\0&0\end{pmatrix}
+\begin{pmatrix}0&0\\0&i\end{pmatrix}
=\begin{pmatrix}1&0\\0&i\end{pmatrix}.
$$
So, we obtain that
$$
(B^*)^{-1}I_0(0)^{-1}B^{-1}
=\big(B I_0(0) B^*\big)^{-1}
=2\begin{pmatrix}1&0\\0&i\end{pmatrix}^{-1}
=2\begin{pmatrix}1&0\\0&-i\end{pmatrix},
$$
which is equivalent to
$
\big(I_0(0)|_\L\big)^{-1}
=2B^*\left(\begin{smallmatrix}1&0\\0&-i\end{smallmatrix}\right)B
$.
Using this equality, we can then show that all the matrix coefficients of
$\P_N^0\v I_0(0)^{-1}\v\;\!\P_{N/2}^0$ are zero. Namely, for any
$j,j'\in\{1,\dots,N\}$ we have
\begin{align*}
\left\langle\xi^0_j,\big(\P_N^0 \v I_0(0)^{-1}\v\;\!\P_{N/2}^0\big)
\xi^0_{j'}\right\rangle_{\C^N}
&=\delta_{j,N}\;\!\delta_{j',N/2}\left\langle\v\xi_N^0,
I_0(0)^{-1}\v\xi^0_{N/2}\right\rangle_{\C^N}\\
&=2\;\!\delta_{j,N}\;\!\delta_{j',N/2}\left\langle B\v\xi_N^0,
\begin{pmatrix}1&0\\0&-i\end{pmatrix}B\v\xi^0_{N/2}\right\rangle_{\C^2}\\
&=2\;\!\delta_{j,N}\;\!\delta_{j',N/2}\left\langle\begin{pmatrix}1\\0\end{pmatrix},
\begin{pmatrix}1&0\\0&-i\end{pmatrix}\begin{pmatrix}0\\1\end{pmatrix}
\right\rangle_{\C^2}\\
&=0.
\end{align*}

(ii)$\Rightarrow$(i) Suppose now that $\v\xi^0_N$ and $\v\xi^0_{N/2}$ are linearly
dependent. Then there exists $\alpha\in\C^*$ such that
$\v\xi^0_{N/2}=\alpha\v\xi^0_N$, and
$$
2\;\!I_0(0)
=\v\;\!\P^0_N\v+i\;\!\v\;\!\P^0_{N/2}\v
=\big(1+i\;\!|\alpha|^2\big)\v\;\!\P_N^0\v.
$$
Defining
$
\big(\big|\xi^0_N\big\rangle\big\langle\xi^0_{N/2}\big|\big)\xi
:=\big\langle\xi^0_{N/2},\xi\big\rangle_{\C^N}\;\!\xi^0_N
$
for any $\xi\in\C^N$, we thus obtain that
\begin{align*}
\P_N^0 \v I_0(0)^{-1}\v\;\!\P_{N/2}^0
&=\tfrac1{2(1+i\;\!|\alpha|^2)}\;\!\P_N^0 \v\big(\v\;\!\P_N^0\P_N^0\v\big)^{-1}
\v\;\!\P_{N/2}^0\\
&=\tfrac1{2(1+i\;\!|\alpha|^2)}\;\!\big(\v\;\!\P_N^0)^{-1}\v\;\!\P_{N/2}^0\\
&=\tfrac\alpha{2(1+i\;\!|\alpha|^2)}\big|\xi^0_N\big\rangle
\big\langle\xi^0_{N/2}\big|\\
&\ne0.
\end{align*}
(i)$\Leftrightarrow$(iii) This equivalence can be shown as the equivalence
(i)$\Leftrightarrow$(ii).
\end{proof}

\begin{Remark}\label{rem_components_V}
In general, the vectors $\v\xi^0_N$ and $\v\xi^0_{N/2}$ are linearly independent.
Indeed, an inspection using \eqref{eq_vectors} shows that $\v\xi^0_N$ and
$\v\xi^0_{N/2}$ are linearly dependent if only if the matrix $\v$ is of the special
form \eqref{eq_spf}. Accordingly, we shall call \emph{degenerate case} the very
exceptional case where $\theta=0$, $N\in2\N$, and $\v\xi_N^0$ and $\v\xi^0_{N/2}$ are
linearly dependent.
\end{Remark}

Using Lemma \ref{lemma_independent} and results of the previous section we can prove
the compactness of the operator defined by the remainder term \eqref{eq_remainder}
when the vectors $\v\xi^0_N$ and $\v\xi^0_{N/2}$ are linearly independent:

\begin{Proposition}\label{prop_exceptional}
If  $\theta=0$, $N\in2\N$, and the vectors $\v\xi^0_N$ and $\v\xi^0_{N/2}$ are
linearly independent, then the operator defined by \eqref{eq_remainder} extends
continuously to a compact operator in $\Hrond^\theta$.
\end{Proposition}

\begin{proof}
We know from the proof of Proposition \ref{prop_NUS2} that \eqref{eq_remainder} can be
rewritten as
$$
\left\langle k^\theta\xi,\zeta\right\rangle_{\oplus_{j=1}^N
\ltwo(I_j^\theta;\P_j^\theta\C^N)},\quad\xi,\zeta\in\Drond^\theta,
$$
with $k^\theta:\Hrond^\theta\to\oplus_{j=1}^N\ltwo(I_j^\theta;\P_j^\theta\C^N)$ given
by
$$
\big(k^\theta\xi\big)_j
:=\pi^{-1}\sum_{j'=1}^N(\vartheta_{j,j'}^\theta)^*
N_{j,j'}^\theta(1_{j,j'}^\theta)^*\xi_{j'}.
$$
Furthermore, we know that each operator
$(\vartheta_{j,j'}^\theta)^*N_{j,j'}^\theta(1_{j,j'}^\theta)^*$ is compact except for
the values $\theta=0$, $N\in2\N$, $(j,j')=(N,N/2)$ (first case) or $(j,j')=(N/2,N)$
(second case). So, it is sufficient to prove that the operators
$(\vartheta_{N,N/2}^0)^*N_{N,N/2}^0(1_{N,N/2}^0)^*$ and
$(\vartheta_{N/2,N}^0)^*N_{N/2,N}^0(1_{N/2,N}^0)^*$ are compact too. We give the proof
only for the first operator, since the second operator is similar.

Using the notations of the proof of Lemma \ref{lemma_compact} (with $\alpha=4$) and
the fact that $1_{N,N/2}^0$ is the identity operator, we obtain
\begin{align}
&(\vartheta_{N,N/2}^0)^*N_{N,N/2}^0(1_{N,N/2}^0)^*\nonumber\\
&=\Big\{\Big(U_2^*(1_{(0,4)})^*m(X_+)\varphi(A_+)
(4-X_+)^{-1/4}\eta^\bot(X_+)1_{(0,4)}U_1\Big)^*\otimes1_N\Big\}
N_{N,N/2}^0+K\nonumber\\
&=\Big(U_1^*(1_{(0,4)})^*(4-X_+)^{-1/4}\eta^\bot(X_+)\overline\varphi(A_+)
m(X_+)1_{(0,4)}U_2\otimes1_N\Big)N_{N,N/2}^0+K\nonumber\\
&=\Big(U_1^*(1_{(0,4)})^*(4-X_+)^{-1/4}\eta^\bot(X_+)\overline\varphi(A_+)
m(X_+)\otimes1_N\Big)\tilde n_{N,N/2}^0(X_+)\big(1_{(0,4)}U_2\otimes1_N\big)+K,
\label{eq_first_term}
\end{align}
with $K$ a compact operator and $\tilde n_{N,N/2}^0:(0,\infty)\to\B(\P^0_N\C^N)$ the
continuous function given by (see \eqref{eq_def_n})
$$
\tilde n_{N,N/2}^0(x):=
\begin{cases}
n_{N,N/2}^0(-x) & \hbox{$x\in(0,4)$,}\\
n_{N,N/2}^0(-4) & \hbox{$x\ge4$.}
\end{cases}
$$
Since $m(x)=0$ for $x\ge4$ and since
$$
\tilde n_{N,N/2}^0(\kappa^2)
=\beta^0_N(-\kappa^2)^{-2}\P^0_N\v\M^0(0,\kappa)\v\;\!\P^0_{N/2}
=\O(\kappa),\quad\kappa>0,
$$
due to \eqref{eq_Wuhan3} and Lemma \ref{lemma_independent}, the continuous function
$$
(0,\infty)\ni x\mapsto
\big(m(x)\otimes1_N\big)\tilde n_{N,N/2}^0(x)\in\P^0_N\C^N
$$
vanishes at $x=0$ and for $x\ge4$. Therefore, one can reproduce the argument at the
end of proof of Lemma \ref{lemma_compact} to conclude that the first term in
\eqref{eq_first_term} is compact.
\end{proof}

\subsection{New formula for the wave operators}\label{sec_new_formula}

Using the results obtained in the previous sections, we can finally derive a new
formula for the wave operators $W_\pm^\theta$. We recall that the case where
$\theta=0$, $N\in2\N$, and the vectors $\v\xi_N^0$ and $\v\xi^0_{N/2}$ are linearly
dependent, is referred as the degenerate case (see  Remark \ref{rem_components_V}). By
combining the results of equations \eqref{eq_form_leading}-\eqref{eq_compact_leading}
and Propositions \ref{prop_NUS2} \& \ref{prop_exceptional}, we get for any
$\theta\in[0,2\pi]$ the equality
$$
\F^\theta(W_-^\theta-1)(\F^\theta)^*
=\tfrac12(\V^{\theta})^*\big\{\big(1-\tanh(\pi D)-i\cosh(\pi D)^{-1}\tanh(X)\big)
\otimes1_N\big\}\V^\theta\big(S^\theta(X^\theta)-1\big)+K^\theta,
$$
with $K^\theta\in\K(\Hrond^\theta)$ in the nondegenerate cases, and
$K^0\in\B(\Hrond^0)$ in the degenerate case.

In order to obtain an expression for the operator $(W_-^\theta-1)$ alone, we introduce
the operators in $\h$
$$
\Xt:=(\V^{\theta}\F^\theta)^*(X\otimes1_N)\;\!\V^{\theta}\F^\theta
\quad\hbox{and}\quad
\Dt:=(\V^{\theta}\F^\theta)^*(D\otimes1_N)\;\!\V^{\theta}\F^\theta
$$
with domains $\dom(\Xt):=(\V^{\theta}\F^\theta)^*\;\!\dom(X\otimes1_N)$ and
$\dom(\Dt):=(\V^{\theta}\F^\theta)^*\;\!\dom(D\otimes1_N)$. These operators are
self-adjoint, satisfy the canonical commutation relation (because $X$ and $D$ satisfy
it) and are independent of the variable $\theta$. Namely,
$$
(\Xt\;\!g)(\omega)=\arctanh(\cos(\omega))\;\!g(\omega)
\quad\hbox{$g\in\dom(\Xt)$, a.e. $\omega\in[0,\pi)$,}
$$
and $\Dt=\Ut^*\Xt\;\!\Ut$ with
$\Ut:=(\V^{\theta}\F^\theta)^*(\F\otimes1_N)\;\!\V^{\theta}\F^\theta$ unitary and
independent of $\theta$, and $\F\in\B\big(\ltwo(\R)\big)$ the Fourier transform on
$\R$.

Using the operators $\Xt$ and $\Dt$, we thus obtain the desired formula for the wave
operator $W_-^\theta$ (and thus also for $W_+^\theta$ if we use the relation
$W_+^\theta=W_-^\theta(S^\theta)^*$)\;\!:

\begin{Theorem}\label{thm_formula_theta}
For any $\theta\in[0,2\pi]$, one has the equality
$$
W_-^\theta-1
=\tfrac12\big(1-\tanh(\pi\Dt)-i\cosh(\pi\Dt)^{-1}\tanh(\Xt)\big)(S^\theta-1)
+\Kt^\theta,
$$
with $\Kt^\theta:=(\F^\theta)^*K^\theta\F^\theta\in\K(\h)$ in the nondegenerate cases,
and $\Kt^0:=(\F^0)^*K^0\F^0\in\B(\h)$ in the degenerate case.
\end{Theorem}

\begin{Remark}\label{remark_future}
The result of Theorem \ref{thm_formula_theta} is weaker in the degenerate case, when
we only prove that $\Kt^0\in\B(\h)$. However, there is plenty of space left between the
set of compact operators $\K(\h)$ and the set of bounded operators $\B(\h)$. In a
second paper, we plan to show that, even in the degenerate case, the remainder
term $\Kt^0$ is small in a suitable sense compared to the leading term
$\tfrac12\big(1-\tanh(\pi\Dt)-i\cosh(\pi\Dt)^{-1}\tanh(\Xt)\big)(S^0-1)$. This will be
achieved by showing that $\Kt^0$ belongs to a $C^*$-algebra bigger than the set of
compact operators, but smaller than the set of all bounded operators.
\end{Remark}

Finally, we derive a formula for the wave operators
$W_\pm=\slim_{t\to\pm\infty}\e^{itH}\e^{-itH_0}$ for the initial pair of Hamiltonians
$(H,H_0)$. As explained in the last part of Section \ref{sec_intro}, the wave
operators $W_\pm$ exist and have same range. In addition, both the wave operators
$W_\pm$ and the scattering $S=(W_+)^*W_-$ admit direct integral decompositions
$$
\G\;\! W_\pm\;\!\G^*=\int_{[0,2\pi]}^\oplus W_\pm^\theta\;\!\tfrac{\d\theta}{2\pi}
\quad\hbox{and}\quad
\G\;\!S\;\!\G^*=\int_{[0,2\pi]}^\oplus S^\theta\;\!\tfrac{\d\theta}{2\pi}
$$
with $\G:\H\to\int_{[0,2\pi]}^\oplus\h\;\!\tfrac{\d\theta}{2\pi}$ the unitary operator
defined in Section \ref{sec_direct}. Therefore, by collecting the formulas obtained in
Theorem \ref{thm_formula_theta} for $(W_-^\theta-1)$ in each fiber Hilbert space $\h$,
we obtain a formula for $(W_--1)$ in the full direct sum Hilbert space
$\int_{[0,2\pi]}^\oplus\h\;\!\tfrac{\d\theta}{2\pi}$ (and thus also for $W_+$ if we
use the relation $W_+=W_-S^*$)\;\!:

\begin{Theorem}\label{thm_formula}
One has the equality
$$
\G(W_--1)\G^*
=\int_{[0,2\pi]}^\oplus\Big(\tfrac12\big(1-\tanh(\pi\Dt)
-i\cosh(\pi\Dt)^{-1}\tanh(\Xt)\big)(S^\theta-1)+\Kt^\theta\Big)
\tfrac{\d\theta}{2\pi},
$$
with $\Kt^\theta$ as in Theorem \ref{thm_formula_theta}.
\end{Theorem}

\appendix
\section{Appendix}
\setcounter{equation}{0}
\renewcommand{\theequation}{A.\arabic{equation}}

In this Appendix, we provide the proofs of several results which have only been stated
in the main sections of this paper.

\begin{proof}[Proof of Lemma \ref{lemma_sandwich}]
For any $\lambda_\star\in\R$ and $z\in\C\setminus\R$, we have
\begin{align*}
\int_0^\pi\big(2\cos(\omega)+\lambda_\star-z\big)^{-1}\tfrac{\d\omega}\pi
&=\tfrac1{2\pi}\int_0^{2\pi}\big(2\cos(\omega)+\lambda_\star-z\big)^{-1}\d\omega\\
&=\tfrac1{2\pi}\int_0^{2\pi}\e^{i\omega}
\big(\e^{2i\omega}-(z-\lambda_\star)\e^{i\omega}+1\big)^{-1}\d\omega\\
&=\tfrac1{2\pi i}\int_{\S^1}(\varsigma-a_+)^{-1}(\varsigma-a_-)^{-1}\;\!\d\varsigma
\end{align*}
with $a_\pm:=\tfrac12\big(z-\lambda_\star\pm\sqrt{(z-\lambda_\star)^2-4}\big)$.
Therefore, one gets from the residue theorem that
\begin{equation}\label{eq_res_1}
\int_0^\pi\big(2\cos(\omega)+\lambda_\star-z\big)^{-1}\tfrac{\d\omega}\pi
=\begin{cases}
\frac1{a_+-a_-}=\frac1{\sqrt{(z-\lambda_\star)^2-4}} &\hbox{if $a_+\in\D$}\\
-\frac1{a_+-a_-}=-\frac1{\sqrt{(z-\lambda_\star)^2-4}} &\hbox{if $a_-\in\D$.}
\end{cases}
\end{equation}

Now, take $z=\lambda+i\varepsilon$ with $\varepsilon>0$ small enough, and set
$\alpha:=\lambda-\lambda_\star$. Then, in the case $\alpha^2>4$ we obtain that
\begin{align*}
\lim_{\varepsilon\searrow0}a_\pm
&=\tfrac12\lim_{\varepsilon\searrow0}\big(\alpha+i\varepsilon\pm
\sqrt{\alpha^2-4+2i\alpha\varepsilon-\varepsilon^2}\big)\\
&=\tfrac12\lim_{\varepsilon\searrow0}\Big(\alpha+i\varepsilon\pm\sqrt{\alpha^2-4}
\cdot\sqrt{1+\tfrac{2i\alpha\varepsilon}{\alpha^2-4}+\O(\varepsilon^2)}\Big)\\
&=\tfrac12\big(\alpha\pm\sqrt{\alpha^2-4}\cdot\sgn(\alpha)\big)\\
&=\tfrac2{\alpha\mp\sqrt{\alpha^2-4}\cdot\sgn(\alpha)},
\end{align*}
which implies that $a_-\in\D$ if $\alpha^2>4$ and $\varepsilon$ is small enough.
Similarly, in the case $\alpha^2<4$ we obtain that
\begin{align*}
a_\pm
&=\tfrac12\big(\alpha+i\varepsilon\pm
\sqrt{\alpha^2-4+2i\alpha\varepsilon-\varepsilon^2}\big)\\
&=\tfrac12\left(\alpha+i\varepsilon\pm\sqrt{4-\alpha^2}\cdot
\sqrt{-1+\tfrac{2i\alpha\varepsilon}{4-\alpha^2}+\O(\varepsilon^2)}\right)\\
&=\tfrac12\left(\alpha+i\varepsilon\pm i\sqrt{4-\alpha^2}\cdot
\left(1-\tfrac{i\alpha\varepsilon}{4-\alpha^2}+\O(\varepsilon^2)\right)\right)\\
&=\tfrac12\left(\big(\alpha\pm i\sqrt{4-\alpha^2}\big)
\left(1\pm\tfrac\varepsilon{\sqrt{4-\alpha^2}}\right)+\O(\varepsilon^2)\right)
\end{align*}
which implies that $a_-\in\D$ if $\alpha^2<4$ and $\varepsilon$ is small enough.
Putting these formulas for $a_\pm$ in \eqref{eq_res_1}, we get
$$
\lim_{\varepsilon\searrow0}\int_0^\pi
\big(2\cos(\omega)+\lambda_\star-(\lambda+i\varepsilon)\big)^{-1}\tfrac{\d\omega}\pi
=\begin{cases}
\big|(\lambda-\lambda_\star)^2-4\big|^{-1/2} &\hbox{if $\lambda<\lambda_\star-2$}\\
i\;\!\big|(\lambda-\lambda_\star)^2-4\big|^{-1/2} &
\hbox{if $\lambda\in(\lambda_\star-2,\lambda_\star+2)$}\\
-\big|(\lambda-\lambda_\star)^2-4\big|^{-1/2} &\hbox{if $\lambda>\lambda_\star+2$}.
\end{cases}
$$

Finally, the last equation and the formulas for $H_0^\theta$, $G$, $G^*$ imply that
\begin{align*}
GR_0^\theta(\lambda+i\;\!0)G^*
&=\lim_{\varepsilon\searrow0}\v\int_0^\pi
\big(2\cos(\omega)+A^\theta-\lambda-i\varepsilon\big)^{-1}\tfrac{\d\omega}\pi\;\!\v\\
&=\v\lim_{\varepsilon\searrow0}\int_0^\pi\Bigg(2\cos(\omega)
+\sum_{j=1}^N\lambda_j^\theta\;\!\P_j^\theta-\lambda
-i\varepsilon\Bigg)^{-1}\tfrac{\d\omega}\pi\;\!\v\\
&=\sum_{j=1}^N\v\lim_{\varepsilon\searrow0}\int_0^\pi
\big(2\cos(\omega)+\lambda_j^\theta-\lambda-i\varepsilon\big)^{-1}
\tfrac{\d\omega}\pi\;\!\P_j^\theta\v\\
&=\sum_{\{j\mid\lambda<\lambda_j^\theta-2\}}\tfrac{\v\;\!\P_j^\theta\v}
{\beta_j^\theta(\lambda)^2}
+i\sum_{\{j\mid\lambda\in I^\theta_j\}}\tfrac{\v\;\!\P_j^\theta\v}
{\beta_j^\theta(\lambda)^2}
-\sum_{\{j\mid\lambda>\lambda_j^\theta+2\}}\tfrac{\v\;\!\P_j^\theta\v}
{\beta_j^\theta(\lambda)^2},
\end{align*}
which proves the claim.
\end{proof}

\begin{proof}[Proof of Lemma \ref{lemme_com}]
The fact that $S_\ell$ is the orthogonal projection on the kernel of $I_\ell(0)$ and
the relations $S_mS_\ell=S_\ell=S_\ell S_m$ imply that $[S_m,S_\ell]=0$ and
$[I_m(0),S_\ell]=0$. Thus, one has in $\B(\C^N)$ the equalities
\begin{align*}
\big[S_\ell,\big(I_m(\kappa)+S_m\big)^{-1}\big]
&=\big(I_m(\kappa)+S_m\big)^{-1}\big[I_m(\kappa)+S_m,S_\ell\big]
\big(I_m(\kappa)+S_m\big)^{-1}\\
&=\big(I_m(\kappa)+S_m\big)^{-1}\big[I_m(0)+\O(\kappa)+S_m,S_\ell\big]
\big(I_m(\kappa)+S_m\big)^{-1}\\
&=\big(I_m(\kappa)+S_m\big)^{-1}\big[\O(\kappa),S_\ell\big]
\big(I_m(\kappa)+S_m\big)^{-1},
\end{align*}
which imply the claim.
\end{proof}

\begin{proof}[Proof of Lemma \ref{lemma_relations}]
(a) follows from the fact that $S_0$ is the orthogonal projection on
$\ker\big(I_0(0)\big)$ and the fact that both $\re\big(I_0(0)\big)$ and
$\im\big(I_0(0)\big)$ are sums of positive operators. Similarly, (b) follows from the
fact that $S_1$ is the orthogonal projection on the kernel of $I_1(0)$ and the fact
that $\im\big(I_1(0)\big)$ is a sum of positive operators. For (c), recall that $S_2$
is the orthogonal projection on the kernel of $I_2(0)$ and that $\im\big(I_2(0)\big)$
is positive. Thus,
$$
S_2\re\big(I_2(0)\big)S_2=0=S_2\im\big(I_2(0)\big)S_2.
$$
With the notations of \eqref{eq_AB} this implies that the range of the operator
$(A-iB^*B)^{-1}\re\big(M_1(0)\big)S_2$ belongs to both $\ker(A)$ (first equality) and
$\ker(B)$ (second equality). However, since $(A-iB^*B)$ is invertible, the only
element in $\ker(A)\cap\ker(B)$ is the vector $0$. Therefore, we have
$(A-iB^*B)^{-1}\re\big(M_1(0)\big)S_2=0$, and thus
$$
\re\big(M_1(0)\big)S_2=0=S_2\re\big(M_1(0)\big).
$$
Finally, (d) follows from (c), since we know from the proof of Proposition
\ref{Prop_asymp} that $\im\big(M_1(0)\big)S_2=0$.
\end{proof}

\begin{proof}[Proof of Theorem \ref{thm_cont}]
(a) Some lengthy, but direct, computations taking into account the expansion
\eqref{eq_grosse}, the relation $\big(I_\ell(0)+S_\ell\big)^{-1}S_\ell=S_\ell$, the
expansion
\begin{equation}\label{eq_expansion_beta}
\beta_j^\theta(\lambda-\kappa^2)^{-1}
=\beta_j^\theta(\lambda)^{-1}\Big(1+\tfrac{\kappa^2}
{2(\lambda-\lambda_j^\theta)}+\O(\kappa^4)\Big),
\quad\lambda\in I^\theta_j,
\end{equation}
and Lemma \ref{lemma_relations}(b) lead to the equality
\begin{align*}
&\lim_{\kappa\to0}\beta_j^\theta(\lambda-\kappa^2)^{-1}
\P_j^\theta\v\;\!\M^\theta(\lambda,\kappa)\v\;\!\P_{j'}^\theta\;\!
\beta_{j'}^\theta(\lambda-\kappa^2)^{-1}\\
&=\beta_j^\theta(\lambda)^{-1}\P_j^\theta\v\;\!S_0\big(I_1(0)+S_1\big)^{-1}
S_0\v\;\!\P_{j'}^\theta\beta_{j'}^\theta(\lambda)^{-1}\\
&\quad-\beta_j^\theta(\lambda)^{-1}\P_j^\theta\v\big(C_{20}'(0)+S_0C_{21}'(0)\big)
S_2I_3(0)^{-1}S_2\big(C_{20}'(0)+C_{21}'(0)S_0\big)\v\;\!\P_{j'}^\theta\;\!
\beta_{j'}^\theta(\lambda)^{-1}.
\end{align*}
Moreover, Lemmas \ref{lemma_relations}(a) \& \ref{lemma_relations}(d) imply that
\begin{align}
C_{20}(\kappa)
&=\big(I_0(\kappa)+S_0\big)^{-1}
\Bigg[-\sum_{j\in\N_\lambda}\tfrac{\v\;\!\P_j^\theta\v}{\vartheta_j(\kappa)}
+\kappa\;\!M_1(\kappa),S_2\Bigg]\big(I_0(\kappa)+S_0\big)^{-1}\nonumber\\
&=\kappa\;\!\big(I_0(\kappa)+S_0\big)^{-1}\big[M_1(0),S_2\big]
\big(I_0(\kappa)+S_0\big)^{-1}+\Oas(\kappa^3)\nonumber\\
&=\Oas(\kappa^3),\label{eq_C20}
\end{align}
and Lemma \ref{lemma_relations}(d) and the expansion \eqref{eq_23} imply that
\begin{align*}
C_{21}(\kappa)
&=\big(I_1(\kappa)+S_1\big)^{-1}\big[S_0M_1(0)S_0+\kappa\;\!M_2(\kappa),S_2\big]
\big(I_1(\kappa)+S_1\big)^{-1}\nonumber\\
&=\kappa\;\!\big(I_1(\kappa)+S_0\big)^{-1}\big[M_2(\kappa),S_2\big]
\big(I_1(\kappa)+S_0\big)^{-1}\nonumber\\
&=\kappa\;\!\big(I_1(\kappa)+S_0\big)^{-1}
\big[-S_0M_1(0)\big(I_0(0)+S_0\big)^{-1}M_1(0)S_0,S_2\big]
\big(I_1(\kappa)+S_0\big)^{-1}+\Oas(\kappa^2)\nonumber\\
&=\Oas(\kappa^2).
\end{align*}
Therefore, one has $C_{20}'(0)=C_{21}'(0)=0$, and thus
$$
\lim_{\kappa\to0}\beta_j^\theta(\lambda-\kappa^2)^{-1}\P_j^\theta\v\;\!
\M^\theta(\lambda,\kappa)\v\;\!\P_{j'}^\theta\;\!
\beta_{j'}^\theta(\lambda-\kappa^2)^{-1}
=\beta_j^\theta(\lambda)^{-1}\P_j^\theta\v\;\!S_0\big(I_1(0)+S_1\big)^{-1}
S_0\v\;\!\P_{j'}^\theta\beta_{j'}^\theta(\lambda)^{-1}.
$$
Since
\begin{equation}\label{eq_start}
S^\theta(\lambda-\kappa^2)_{jj'}-\delta_{jj'}
=-2i\;\!\beta_j^\theta(\lambda-\kappa^2)^{-1}\P_j^\theta\v\;\!\M^\theta(\lambda,\kappa)
\v\;\!\P_{j'}^\theta\beta_{j'}^\theta(\lambda-\kappa^2)^{-1},
\end{equation}
this proves the claim.

(b.1) We first consider the case $\lambda=\lambda_{j'}^\theta-2$,
$\lambda>\lambda_j^\theta-2$ (the case $\lambda=\lambda_j^\theta-2$,
$\lambda>\lambda_{j'}^\theta-2$ is not presented since it is similar). An inspection
of the expansion \eqref{eq_grosse} taking into account the relations
$
\big(I_\ell(\kappa)+S_\ell\big)^{-1}=\big(I_\ell(0)+S_\ell\big)^{-1}+\Oas(\kappa)
$
and $\big(I_\ell(0)+S_\ell\big)^{-1}S_\ell=S_\ell$ leads to the equation
\begin{align*}
&\beta_j^\theta(\lambda-\kappa^2)^{-1}\P_j^\theta\v\;\!\M^\theta(\lambda,\kappa)
\v\;\!\P_{j'}^\theta\beta_{j'}^\theta(\lambda-\kappa^2)^{-1}\\
&=\beta_j^\theta(\lambda-\kappa^2)^{-1}\P_j^\theta\v\;\!\Big\{
\Oas(\kappa)+S_0\big(I_1(\kappa)+S_1\big)^{-1}S_0\\
&\quad+\tfrac1\kappa\big(S_1+\Oas(\kappa)\big)S_1
\big(I_2(\kappa)+S_2\big)^{-1}S_1\big(S_1+\Oas(\kappa)\big)\\
&\quad+\tfrac1{\kappa^2}\Big(\Oas(\kappa^2)+S_2\big(I_0(\kappa)+ S_0\big)^{-1}
\big(I_1(\kappa)+S_1\big)^{-1}\big(I_2(\kappa)+S_2\big)^{-1}-C_{20}(\kappa)\\
&\quad-S_0C_{21}(\kappa)-S_1C_{22}(\kappa)\Big)
S_2 I_3(\kappa)^{-1}S_2\Big(\big(I_2(\kappa)+S_2\big)^{-1}
\big(I_1(\kappa)+S_1\big)^{-1}\big(I_0(\kappa)+S_0\big)^{-1}S_2\\
&\quad +C_{20}(\kappa) +C_{21}(\kappa)S_0
+C_{22}(\kappa)S_1 + \Oas(\kappa^2)\Big)\Big\}\;\!\v\;\!\P_{j'}^\theta
\;\!\beta_{j'}^\theta(\lambda-\kappa^2)^{-1}.
\end{align*}
An application of Lemma \ref{lemma_relations}(a)-(b) to the above equation gives
\begin{align*}
&\beta_j^\theta(\lambda-\kappa^2)^{-1}\P_j^\theta\v\;\!\M^\theta(\lambda,\kappa)
\v\;\!\P_{j'}^\theta\beta_{j'}^\theta(\lambda-\kappa^2)^{-1}\\
&=\beta_j^\theta(\lambda-\kappa^2)^{-1}\P_j^\theta\v
\Big(\Oas(\kappa)-\tfrac1{\kappa^2}\big(\Oas(\kappa^2)+C_{20}(\kappa)
+S_0C_{21}(\kappa)\big)\\
&\quad\cdot S_2 I_3(\kappa)^{-1}S_2\big(\Oas(\kappa^2)+C_{20}(\kappa)\big)\Big)
\v\;\!\P_{j'}^\theta\beta_{j'}^\theta(\lambda-\kappa^2)^{-1}.
\end{align*}
Finally, if one takes into account the expansion
$
\beta_j^\theta(\lambda-\kappa^2)^{-1}=\beta_j^\theta(\lambda)^{-1}+\O(\kappa^2)
$
(see \eqref{eq_expansion_beta}) and the equality
$
\beta_{j'}^\theta(\lambda-\kappa^2)^{-1}=|4\kappa^2+\kappa^4|^{-1/4}
$,
one ends up with
\begin{align*}
&\beta_j^\theta(\lambda-\kappa^2)^{-1}\P_j^\theta\v\;\!\M^\theta(\lambda,\kappa)
\v\;\!\P_{j'}^\theta\beta_{j'}^\theta(\lambda-\kappa^2)^{-1}\\
&=\big(\beta_j^\theta(\lambda)^{-1}+\O(\kappa^2)\big)\P_j^\theta\v
\Big(\Oas(\kappa)-\tfrac1{\kappa^2}\big(\Oas(\kappa^2)+C_{20}(\kappa)
+S_0C_{21}(\kappa)\big)\\
&\quad\cdot S_2 I_3(\kappa)^{-1}S_2\big(\Oas(\kappa^2)+C_{20}(\kappa)\big)\Big)
\v\;\!\P_{j'}^\theta\;\!\big|4\kappa^2+\kappa^4\big|^{-1/4}.
\end{align*}
Since $C_{20}(\kappa)=\Oas(\kappa^3)$ (see \eqref{eq_C20}), one infers that
$
\beta_j^\theta(\lambda-\kappa^2)^{-1}\P_j^\theta\v\;\!\M^\theta(\lambda,\kappa)
\v\;\!\P_{j'}^\theta\beta_{j'}^\theta(\lambda-\kappa^2)^{-1}
$
vanishes as $\kappa\to0$, and thus that
$\lim_{\kappa\to0}S^\theta(\lambda-\kappa^2)_{jj'}=0$ due to \eqref{eq_start}.

(b.2) We now consider the case $\lambda_j^\theta-2=\lambda=\lambda_{j'}^\theta-2$. An
inspection of \eqref{eq_grosse} taking into account the relation
$\big(I_\ell(\kappa)+S_\ell\big)^{-1}=\big(I_\ell(0)+S_\ell\big)^{-1}+\Oas(\kappa)$,
the relation $\big(I_\ell(0)+S_\ell\big)^{-1}S_\ell=S_\ell$ and Lemma
\ref{lemma_relations}(a) leads to the equation
\begin{align}
&\beta_j^\theta(\lambda-\kappa^2)^{-1}\P_j^\theta\v\;\!\M^\theta(\lambda,\kappa)
\v\;\!\P_{j'}^\theta\beta_{j'}^\theta(\lambda-\kappa^2)^{-1}\nonumber\\
&=\beta_j^\theta(\lambda-\kappa^2)^{-1}\P_j^\theta\v\Big(\Oas(\kappa^2)
+\kappa\big(I_0(\kappa)+S_0\big)^{-1}-\tfrac1\kappa C_{10}(\kappa)S_1
\big(I_2(\kappa)+S_2\big)^{-1}S_1C_{10}(\kappa)\nonumber\\
&\quad-\tfrac1{\kappa^2}\big(\Oas(\kappa^2)+C_{20}(\kappa)\big)S_2I_3(\kappa)^{-1}
S_2\big(\Oas(\kappa^2)+C_{20}(\kappa)\big)\Big)\v\;\!\P_{j'}^\theta
\beta_{j'}^\theta(\lambda-\kappa^2)^{-1}.\label{eq_b2}
\end{align}
Therefore, since
$
\beta_j^\theta(\lambda-\kappa^2)^{-1}
=\beta_{j'}^\theta(\lambda-\kappa^2)^{-1}
=|4\kappa^2+\kappa^4|^{-1/4},
$
$C_{20}(\kappa)\in\Oas(\kappa^3)$, and $i\kappa\in(0,\varepsilon)$, one obtains that
\begin{align*}
&\lim_{\kappa\to0}\beta_j^\theta(\lambda-\kappa^2)^{-1}\P_j^\theta\v\;\!
\M^\theta(\lambda,\kappa)\v\;\!\P_{j'}^\theta
\beta_{j'}^\theta(\lambda-\kappa^2)^{-1}\\
&=\lim_{\kappa\to0}\big|4\kappa^2+\kappa^4\big|^{-1/2}\P_j^\theta\v\Big(\Oas(\kappa^2)
+\kappa\big(I_0(\kappa)+S_0\big)^{-1}\\
&\quad-\tfrac1\kappa C_{10}(\kappa)S_1
\big(I_2(\kappa)+S_2\big)^{-1}S_1C_{10}(\kappa)\Big)\v\;\!\P_{j'}^\theta\\
&=-\tfrac i2\;\!\P_j^\theta\v\big(I_0(0)+S_0\big)^{-1}\v\;\!\P_{j'}^\theta
+\tfrac i2\;\!\P_j^\theta\v\;\!C_{10}'(0)S_1\big(I_2(0)+S_2\big)^{-1}S_1C_{10}'(0)
\v\;\!\P_{j'}^\theta,
\end{align*}
and thus that
$$
\lim_{\kappa\to0}S^\theta(\lambda-\kappa^2)_{jj'}
=\delta_{jj'}-\P_j^\theta\v\big(I_0(0)+S_0\big)^{-1}\v\;\!\P_{j'}^\theta
+\P_j^\theta\v\;\!C_{10}'(0)S_1\big(I_2(0)+S_2\big)^{-1}S_1C_{10}'(0)
\v\;\!\P_{j'}^\theta
$$
due to \eqref{eq_start}.

(c) The proof is similar to that of (b) except for the last part. Indeed, we now have
$\kappa\in(0,\varepsilon)$ instead of $i\kappa\in(0,\varepsilon)$. Thus \eqref{eq_b2}
implies that
\begin{align*}
&\lim_{\kappa\to0}\beta_j^\theta(\lambda-\kappa^2)^{-1}\P_j^\theta\v\;\!
\M^\theta(\lambda,\kappa)\v\;\!\P_{j'}^\theta
\beta_{j'}^\theta(\lambda-\kappa^2)^{-1}\\
&=\tfrac12\;\!\P_j^\theta\v\big(I_0(0)+S_0\big)^{-1}\v\;\!\P_{j'}^\theta
-\tfrac12\;\!\P_j^\theta\v\;\!C_{10}'(0)S_1\big(I_2(0)+S_2\big)^{-1}S_1C_{10}'(0)
\v\;\!\P_{j'}^\theta,
\end{align*}
and
$$
\lim_{\kappa\to0}S^\theta(\lambda-\kappa^2)_{jj'}
=\delta_{jj'}-i\;\!\P_j^\theta\v\big(I_0(0)+S_0\big)^{-1}\v\;\!\P_{j'}^\theta
+i\;\!\P_j^\theta\v\;\!C_{10}'(0)S_1\big(I_2(0)+S_2\big)^{-1}S_1C_{10}'(0)
\v\;\!\P_{j'}^\theta.\qedhere
$$
\end{proof}

\begin{proof}[Proof of Theorem \ref{thm_cont_bis}]
We know from \eqref{eq_expansion_2} that
$$
\M^\theta(\lambda,\kappa)
=\big(J_0(\kappa)+S\big)^{-1}+\tfrac1{\kappa^2}\big(J_0(\kappa)+S\big)^{-1}S
J_1(\kappa)^{-1}S\big(J_0(\kappa)+S\big)^{-1},
$$
with $S$ the orthogonal projection on the kernel of the operator
$$
T_0=\u+\sum_{\{j\mid\lambda<\lambda_j^\theta-2\}}\tfrac{\v\;\!\P_j^\theta\v}
{\beta_j^\theta(\lambda)^2}
+i\sum_{\{j\mid\lambda\in I^\theta_j\}}\tfrac{\v\;\!\P_j^\theta\v}
{\beta_j^\theta(\lambda)^2}
-\sum_{\{j\mid\lambda>\lambda_j^\theta+2\}}\tfrac{\v\;\!\P_j^\theta\v}
{\beta_j^\theta(\lambda)^2}.
$$
Now, since $J_0(\kappa)=T_0+\kappa^2T_1(\kappa)$ with $T_1(\kappa)\in\Oas(1)$,
commuting $S$ with $\big(J_0(\kappa)+S\big)^{-1}$ gives
$$
\M^\theta(\lambda,\kappa)
=\big(J_0(\kappa)+S\big)^{-1}
+\tfrac1{\kappa^2}\Big(S\big(J_0(\kappa)+S\big)^{-1}+\Oas(\kappa^2)\Big)
SJ_1(\kappa)^{-1}S\Big(\big(J_0(\kappa)+S\big)^{-1}S+\Oas(\kappa^2)\Big),
$$
and an application of \cite[Lemma~2.5]{RT16} shows that
$\P_j^\theta\v\;\!S=0=S\v\;\!\P_j^\theta$ for each $j\in\{1,\dots,N\}$ such that
$\lambda\in I^\theta_j$. These relations, together with \eqref{eq_start}, imply the
equality \eqref{eq_S_vp}.
\end{proof}

\begin{proof}[Proof of Lemma \ref{lemma_Pi}]
For any $\xi\in C^\infty_{\rm c}(I_j^\theta)$ and $\lambda\in I_j^\theta$, one has
$$
\lim_{\varepsilon\searrow0}\big(\Theta_{j,\varepsilon}^\theta\xi\big)(\lambda)
=\tfrac i{2\pi}\pv\int_{I_j^\theta}\tfrac1{\mu-\lambda}\;\!\beta^\theta_j(\lambda)
\;\!\beta_j^\theta(\mu)^{-1}\xi(\mu)\;\!\d\mu+\tfrac12\xi(\lambda)
$$
with $\pv$ the symbol for the usual principal value. With some changes of variables,
it follows that for $f\in C_{\rm c}^\infty(\R)$ and $s\in\R$
$$
\lim_{\varepsilon\searrow0}
\big(\V_j^\theta\Theta_{j,\varepsilon}^\theta(\V_j^{\theta})^*f\big)(s)
=\tfrac i{2\pi}\pv\int_\R\tfrac{\cosh(t)^{1/2}}{\cosh(s)^{1/2}\sinh(t-s)}\;\!f(t)
\;\!\d t+\tfrac12f(s).
$$
Now, one has the identity
\begin{align*}
&\pv\int_\R\tfrac{\e^{t/2}+\e^{-t/2}}{(\e^{s/2}+\e^{-s/2})\sinh(t-s)}\;\!f(t)\;\!\d t\\
&=\tfrac12\pv\int_\R\tfrac{\e^{s/2}}{\e^{s/2}+\e^{-s/2}}
\left(\tfrac1{\sinh((t-s)/2)}+\tfrac1{\cosh((t-s)/2)}\right)f(t)\;\!\d t\\
&\quad+\tfrac12\pv\int_\R\tfrac{\e^{-s/2}}{\e^{s/2}+\e^{-s/2}}
\left(\tfrac1{\sinh((t-s)/2)}-\tfrac1{\cosh((t-s)/2)}\right)f(t)\;\!\d t.
\end{align*}
Therefore, one obtains that
\begin{align*}
&\lim_{\varepsilon\searrow0}
\big(\V_j^\theta\Theta_{j,\varepsilon}^\theta(\V_j^{\theta})^*f\big)(s)\\
&=\tfrac{i\e^{s/2}b_+(s)}{4\pi(\e^{s/2}+\e^{-s/2})}\;\!\pv\int_\R
\left(\tfrac1{\sinh((t-s)/2)}+\tfrac1{\cosh((t-s)/2)}\right)
\big(b_+^{-1}f\big)(t)\;\!\d t\\
&\quad+\tfrac{i\e^{-s/2}b_+(s)}{4\pi(\e^{s/2}+\e^{-s/2})}\;\!\pv\int_\R
\left(\tfrac1{\sinh((t-s)/2)}-\tfrac1{\cosh((t-s)/2)}\right)
\big(b_+^{-1}f\big)(t)\;\!\d t+\tfrac12f(s)\\
&=\tfrac i{4\pi}\;\!b_+(s)\cdot\pv\int_\R\tfrac1{\sinh((t-s)/2)}
\big(b_+^{-1}f\big)(t)\;\!\d t\\
&\quad+\tfrac i{4\pi}\;\!b_-(s)\cdot\pv\int_\R\tfrac1{\cosh((t-s)/2)}
\big(b_+^{-1}f\big)(t)\;\!\d t+\tfrac12f(s)\\
&=-\tfrac i{4\pi}\;\!b_+(s)\cdot\pv\int_\R\csch\big((s-t)/2\big)
\big(b_+^{-1}f\big)(t)\;\!\d t\\
&\quad+\tfrac i{4\pi}\;\!b_-(s)\cdot\pv\int_\R\sech\big((s-t)/2\big)
\big(b_+^{-1}f\big)(t)\;\!\d t+\tfrac12f(s)\\
&=-\tfrac12\big(b_+(X)\tanh(\pi D)b_+(X)^{-1}f\big)(s)
+\tfrac i2\big(b_-(X)\cosh(\pi D)^{-1}b_+(X)^{-1}f\big)(s)+\tfrac12f(s),
\end{align*}
where in the last equality we have used the formulas for the Fourier transform of the
functions $s\mapsto\csch(s/2)$ and $s\mapsto\sech(s/2)$ (see
\cite[Table 20.1]{Jef04}). This concludes the proof of \eqref{eq_kernel}.
\end{proof}

\begin{proof}[Proof of Lemma \ref{lemma_continuity}]
As a first observation, we note that there are two possibilities: either
$\lambda_{j'}^\theta<\lambda_j^\theta$, or $\lambda_j^\theta<\lambda_{j'}^\theta$. If
$\lambda_{j'}^\theta<\lambda_j^\theta$, then
$I^\theta_{j'}\setminus I_j^\theta=(\lambda_{j'}^\theta-2,\lambda_j^\theta-2)$, and we
say that we are in the generic case if $\lambda_j^\theta-2<\lambda_{j'}^\theta+2$ and
in the exceptional case if $\lambda_j^\theta-2=\lambda_{j'}^\theta+2$. On the other
hand, if $\lambda_j^\theta<\lambda_{j'}^\theta$, then
$I^\theta_{j'}\setminus I_j^\theta=(\lambda_j^\theta+2,\lambda_{j'}^\theta+2)$, and we
say that we are in the generic case if $\lambda_{j'}^\theta-2<\lambda_j^\theta+2$ and
in the exceptional case if $\lambda_{j'}^\theta-2=\lambda_j^\theta+2$ (see
\eqref{eq_spec_H^theta_0}). We present below only in the case
$\lambda_{j'}^\theta<\lambda_j^\theta$, since the other case is similar.

Since the function \eqref{eq_def_function} is continuous on
$
[\lambda_{j'}^\theta-2,\lambda_j^\theta-2]
\setminus\big(\T^\theta\cup\sigma_{\rm p}(H^\theta)\big)
$,
we only have to check that the function admits limits in $\B(\C^N)$ as
$
\lambda\to\lambda_\star\in[\lambda_{j'}^\theta-2,\lambda_j^\theta-2]
\cap\big(\T^\theta\cup\sigma_{\rm p}(H^\theta)\big)
$.
However, in order to use the asymptotic expansions of Proposition \ref{Prop_asymp},
we consider values $\lambda-\kappa^2\in\C$ with
$\lambda\in\big(\T^\theta\cup\sigma_{\rm p}(H^\theta)\big)$ and $\kappa\to0$ in a
suitable domain in $\C$ of diameter $\varepsilon>0$. Namely, we treat the three
following possible cases: when $\lambda=\lambda_{j'}^\theta-2$ and
$i\kappa\in(0,\varepsilon)$ (case 1), when $\lambda=\lambda_j^\theta-2$ and
$\kappa\in(0,\varepsilon)$ (case 2), and when
$
\lambda\in(\lambda_{j'}^\theta-2,\lambda_j^\theta-2)
\cap\big(\T^\theta\cup\sigma_{\rm p}(H^\theta)\big)
$
and $\kappa\in(0,\varepsilon)$ or $i\kappa\in(0,\varepsilon)$ (case 3). In each case,
we can choose $\varepsilon>0$ small enough so that
$
\{z\in\C\mid|z-\lambda|<\varepsilon\}
\cap\big(\T^\theta\cup\sigma_{\rm p}(H^\theta)\big)=\{\lambda\}
$
because $\T^\theta$ is discrete and $\sigma_{\rm p}(H^\theta)$ has no accumulation
point (see Remark \ref{remark_no_accumu}(b)).

(i) First, assume that $\lambda\in\sigma_{\rm p}(H^\theta)\setminus\T^\theta$ and let
$\kappa\in(0,\varepsilon)$ or $i\kappa\in(0,\varepsilon)$ with $\varepsilon>0$ small
enough. Then, we know from \eqref{eq_expansion_2} that
$$
\P_j^\theta\v\;\!\M^\theta(\lambda,\kappa)\v\;\!\P_{j'}^\theta
=\P_j^\theta\v\big(J_0(\kappa)+S\big)^{-1}\v\;\!\P_{j'}^\theta
+\tfrac1{\kappa^2}\P_j^\theta\v\big(J_0(\kappa)+S\big)^{-1}SJ_1(\kappa)^{-1}S
\big(J_0(\kappa)+S\big)^{-1}\v\;\!\P_{j'}^\theta
$$
with $S$, $J_0(\kappa)$ and $J_1(\kappa)$ as in point (ii) of the proof of Proposition
\ref{Prop_asymp}. Furthermore, the definitions of $S$ and $J_0(\kappa)$ imply that
$[S,J_0(\kappa)]\in\Oas(\kappa^2)$, and Lemma \ref{lemma_relations}(b) (applied with
$S$ instead of $S_1$) implies that $S\v\;\!\P_{j'}^\theta=0$. Therefore,
\begin{align*}
& \P_j^\theta\v\;\!\M^\theta(\lambda,\kappa)\v\;\!\P_{j'}^\theta\\
&=\Oas(1)+\tfrac1{\kappa^2}\;\!\P_j^\theta\v\big(J_0(\kappa)+S\big)^{-1}
SJ_1(\kappa)^{-1}
S\;\!\big\{\big(J_0(\kappa)+S\big)^{-1}S+\Oas(\kappa^2)\big\}\v\;\!\P_{j'}^\theta\\
&=\Oas(1).
\end{align*}
Since
$
\lim_{\kappa\to0}\beta_j^\theta (\lambda-\kappa^2)^{-2}
=|(\lambda-\lambda_j^\theta)^2-4|^{-1/2}
<\infty
$
for each $\lambda\in\sigma_{\rm p}(H^\theta)\setminus\T^\theta$, we thus infer that
the function \eqref{eq_def_function} (with $\lambda$ replaced by $\lambda-\kappa^2$)
admits a limit in $\B(\C^N)$ as $\kappa\to0$.

(ii) Now, assume that
$\lambda\in[\lambda_{j'}^\theta-2,\lambda_j^\theta-2]\cap\T^\theta$, and consider the
three above cases simultaneously. For this, we recall that $i\kappa\in(0,\varepsilon)$
in case 1, $\kappa\in(0,\varepsilon)$ in case 2, and $\kappa\in(0,\varepsilon)$ or
$i\kappa\in(0,\varepsilon)$ in case 3. Also, we note that the factor
$\beta_j^\theta(\lambda-\kappa^2)^{-2}$ does not play any role in cases 1 and 3, but
gives a singularity of order $|\kappa|^{-1}$ in case 2.

In the expansion \eqref{eq_grosse}, the first term (the one with prefactor $\kappa$)
admits a limit in $\B(\C^N)$ as $\kappa\to0$, even in case 2.

For the second term (the one with no prefactor) only case 2 requires a special
attention: in this case, the existence of the limit as $\kappa\to0$ follows from the
inclusion $C_{00}(\kappa)\in\Oas(\kappa)$ and the equality $\P_j^\theta\v S_0=0$,
which holds by Lemma \ref{lemma_relations}(a).

For the third term (the one with prefactor $1/\kappa$), in case 1 it is sufficient to
observe that $C_{00}(\kappa),C_{10}(\kappa)\in\Oas(\kappa)$ and that
$S_1\v\;\!\P_{j'}^\theta=0$ by Lemma \ref{lemma_relations}(a), and in case 3 it is
sufficient to observe that $C_{00}(\kappa),C_{10}(\kappa)\in\Oas(\kappa)$ and that
$S_1\v\;\!\P_{j'}^\theta=0$ by Lemma \ref{lemma_relations}(b). On the other hand, for
case 2, one must take into account the inclusions
$C_{00}(\kappa),C_{10}(\kappa)\in\Oas(\kappa)$, the equality $\P_j^\theta\v\;\!S_1=0$
given by Lemma \ref{lemma_relations}(a), and the equality $S_1\v\;\!\P_{j'}^\theta=0$
given by Lemma \ref{lemma_relations}(b) in the generic case, or by Lemma
\ref{lemma_relations}(a) in the exceptional case.

For the fourth term (the one with prefactor $1/\kappa^2$), in cases 1 and 3, it is
sufficient to recall that $C_{20}(\kappa)\in\Oas(\kappa^3)$,
$C_{21}(\kappa)\in\Oas(\kappa^2)$, and that
$S_2\v\;\!\P_{j'}^\theta=0=S_1\v\;\!\P_{j'}^\theta$. On the other hand, in case 2, one
must take into account the inclusions $C_{20}(\kappa)\in\Oas(\kappa^3)$,
$C_{21}(\kappa)\in\Oas(\kappa^2)$, the equality $\P_j^\theta\v S_2=0$, and the
equalities $S_2\v\;\!\P_{j'}^\theta=0=S_1\v\;\!\P_{j'}^\theta$ obtained from Lemma
\ref{lemma_relations}(b) in the generic case, or from Lemma \ref{lemma_relations}(a)
in the exceptional case.
\end{proof}

\begin{proof}[Proof of Lemma \ref{lemma_compact}]
We only consider the case where $\lambda_{j'}^\theta<\lambda_j^\theta$ and
$\lambda_j^\theta-2<\lambda_{j'}^\theta+2$, the other case being similar. Let
$\alpha:=\lambda_j^\theta-\lambda_{j'}^\theta\in(0,4)$ and define the two unitary
operators
\begin{align*}
U_1:\ltwo(I_j^\theta)\to\ltwo\big((0,4)\big),
\quad\xi\mapsto\xi\big(\lambda_j^\theta-2+\;\!\cdot\;\!\big),\\
U_2:\ltwo(I_{j'}^\theta\setminus I_j^\theta)\to\ltwo\big((0,\alpha)\big),
\quad\xi\mapsto\xi\big(\lambda_j^\theta-2-\;\!\cdot\;\!\big).
\end{align*}
Then, a straightforward computation gives for $f\in C_{\rm c}^\infty\big((0,4)\big)$
and $x\in(0,\alpha)$ the equality
$$
\big(U_2\;\!\vartheta\;\!U_1^*f\big)(x)
=\int_0^4\tfrac1{x+y}\;\!x^{1/2}(4+x)^{1/2}(\alpha-x)^{-1/4}(4-\alpha+x)^{-1/4}
y^{-1/4}(4-y)^{-1/4}f(y)\;\!\d y.
$$
We will prove the claim by showing that this integral operator on
$C_{\rm c}^\infty\big((0,4)\big)$ extends continuously to a compact operator from
$\ltwo\big((0,4)\big)$ to $\ltwo\big((0,\alpha)\big)$. For simplicity, we keep the
notation $\vartheta$ for the operator $U_2\;\!\vartheta\;\!U_1^*$.

Let $\eta\in C^\infty(\R;\R)$ satisfy $\eta(x)=0$ if $x\le\alpha/4$ and $\eta(x)=1$ if
$x\ge\alpha/2$, and set $\eta^\bot:=1-\eta$. The kernel
$\vartheta(\;\!\cdot,\;\!\cdot\;\!)$ of $\vartheta$ can then be decomposed as
$$
\vartheta(x,y)=\eta^\bot(x)\vartheta(x,y)\eta^\bot(y)
+\eta(x)\vartheta(x,y)\eta^\bot(y)+\eta^\bot(x) \vartheta(x,y)\eta(y)
+\eta(x) \vartheta(x,y)\eta(y)
$$
for each $(x,y)\in(0,\alpha)\times(0,4)$. The last three terms belong to
$\ltwo\big((0,\alpha)\times(0,4)\big)$, and therefore correspond to Hilbert-Schmidt
operators. For the first term, we set
$$
m(x):=
\begin{cases}
x^{1/4}\;\!(4+x)^{1/2}\;\!(\alpha-x)^{-1/4}\;\!(4-\alpha+x)^{-1/4}\eta^\bot(x)
& \hbox{$x\in(0,\alpha)$,}\\
0 & \hbox{$x\ge\alpha$,}
\end{cases}
$$
and observe that $\lim_{\searrow0}m(x)=0$. It follows that
\begin{equation}\label{eq_rem_term}
\eta^\bot(x)\vartheta(x,y)\eta^\bot(y)
=m(x)\cdot\tfrac1{x+y}\;\!x^{1/4}y^{-1/4}\cdot(4-y)^{-1/4}\eta^\bot(y)
\end{equation}
with $m\in C_0\big((0,\infty)\big)$ and
$(0,\infty)\ni y\mapsto(4-y)^{-1/4}\eta^\bot(y)$ bounded. Now, for the central factor
above, we have for any $f\in C_{\rm c}^\infty\big((0,\infty)\big)$ and
$x\in(0,\infty)$ the equalities
\begin{align*}
\int_0^\infty\tfrac1{x+y}\;\!x^{1/4}y^{-1/4}f(y)\;\!\d y
&=\int_0^\infty\tfrac{(x/y)^{1/4}}{(x/y)^{1/2}+(x/y)^{-1/2}}
\;\!(x/y)^{1/2}f(y)\;\!\tfrac{\d y}{x}\\
&=\int_\R\tfrac{\e^{-t/4}}{\e^{t/2}+\e^{-t/2}}\big(V_tf\big)(x)\;\!\d t
\quad\big(y=\e^tx\big)\\
&=\big(\varphi(A_+)f\big)(x)
\end{align*}
with
$$
\varphi(s):=\int_\R\e^{ist}\tfrac{\e^{-t/4}}{\e^{t/2}+\e^{-t/2}}\;\!\d t,\quad s\in\R.
$$
Since $\varphi$ coincides (up to a constant) with the inverse Fourier transform of the
$\lone$-function $t\mapsto\frac{\e^{-t/4}}{\e^{t/2}+\e^{-t/2}}$, we have the inclusion
$\varphi\in C_0(\R)$. Therefore, the operator on $C_{\rm c}^\infty\big((0,4)\big)$
with kernel \eqref{eq_rem_term} extends to the bounded operator from
$\ltwo\big((0,4)\big)$ to $\ltwo\big((0,\alpha)\big):$
\begin{equation}\label{eq_product}
(1_{(0,\alpha)})^*\cdot m(X_+)\varphi(A_+)\cdot(4-X_+)^{-1/4}\eta^\bot(X_+)
\cdot1_{(0,4)},
\end{equation}
with $X_+$ the operator of multiplication by the variable in
$\ltwo\big((0,\infty)\big)$ and $1_{(0,4)}$ (resp. $1_{(0,\alpha)}$) the inclusion of
$\ltwo\big((0,4)\big)$ (resp. $\ltwo\big((0,\alpha)\big)$) into
$\ltwo\big((0,\infty)\big)$. Finally, since $m\in C_0\big((0,\infty)\big)$ and
$\varphi\in C_0(\R)$, the operator $m(X_+)\varphi(A_+)$ is compact in
$\ltwo\big((0,\infty)\big)$ (see for example \cite[Sec.~4.4]{Ric16}), and thus the
operator \eqref{eq_product} is compact from $\ltwo\big((0,4)\big)$ to
$\ltwo\big((0,\alpha)\big)$.
\end{proof}


\end{document}